\documentclass[11pt]{article}
\usepackage{amsmath,amsthm,amssymb,amsfonts,bm}
\usepackage{graphicx,xcolor}

\renewenvironment{equation*}{\gather\textstyle}{\nonumber\endgather}
\DeclareMathOperator*{\plim}{plim}
\usepackage[margin=2.2cm]{geometry}
\usepackage{color}
\usepackage{tabularx,multirow}
\usepackage{enumerate}
\usepackage{graphicx,placeins}
\usepackage{setspace}
\usepackage{lscape}
\newtheorem{theorem}{Theorem}

\newtheorem{lemma}{Lemma}
\newtheorem{assn}{Assumption A}

\newcommand{\vc}{\mbox{\bf vec}}
\newcommand{\diag}{\mbox{\bf diag}}

\newcommand{\inp}{\stackrel{p}{\rightarrow}}
\newcommand{\ind}{\stackrel{d}{\rightarrow}}
\newcommand{\op}{o_p}
\newcommand{\Op}{O_p}

\newcommand{\vbeta}{\bm{\beta}}
\newcommand{\vX}{\bm{X}}

\newcommand{\vXbar}{\bm{\widetilde X}}
\newcommand{\veps}{\bm{\epsilon}}
\newcommand{\vu}{\bm{u}}
\newcommand{\vy}{\bm{y}}
\newcommand{\vx}{\bm{x}}
\newcommand{\vQ}{\bm{Q}}
\newcommand{\va}{\bm{a}}

\newcommand{\vW}{\bm{W}}
\newcommand{\vOmega}{\bm{\Omega}}

\newcommand{\vw}{\bm{w}}
\newcommand{\vz}{\bm{z}}

\newcommand{\tilX}{\widetilde{\bm{X}}}
\newcommand{\vlam}{\bm{\lambda}}
\newcommand{\vP}{\bm{P}}
\newcommand{\vM}{\bm{M}}
\newcommand{\vT}{\bm{T}}
\newcommand{\vI}{\bm I}

\newcommand{\vvc}{\bm c}

\newcommand{\vgam}{\bm{\gamma}}

\newcommand{\vV}{\bm V}
\usepackage{hyperref}

\newcommand{\black}[1]{{\color{black}#1}}

\begin{document}

\title{Change Point Estimation in Panel Data with Time-Varying Individual Effects\thanks{We are grateful to Maurice Bun, Pavel \v{C}\'{i}\v{z}ek, Dick van Dijk, Bas Werker, John Einmahl, Bertrand Melenberg, Frank Kleibergen, Tobias Klein, Paulo Paruolo, as well as to the participants at the Tilburg Seminar in 2016 and at the conferences NESG 2016, IAAE 2016, NASM 2016, ERMAS 2016 and EC2 2017 for useful comments. Furthermore, we would like to kindly thank Junhui Qian and Liangjun Su for making their AGFL code available to us, and Geng Niu and Arthur van Soest for sharing their data.}}
\author{Otilia Boldea\thanks{Tilburg University, Department of Econometrics and Operation Research and CentER, E-mail:o.boldea@uvt.nl.} , Bettina Drepper\thanks{Tilburg University, Department of Econometrics and Operation Research and CentER, E-mail:b.drepper@uvt.nl.} \ and Zhuojiong Gan\thanks{Corresponding author, School of Statistics, Southwestern University of Finance and Economics, Chengdu, China, E-mail: ganzj@swufe.edu.cn.}\,}
\date{\today}
\maketitle

\begin{abstract}

This paper proposes a method for estimating multiple change points in panel data models with unobserved individual effects via ordinary least-squares (OLS). Typically, in this setting, the OLS slope estimators are inconsistent due to the unobserved individual effects bias. As a consequence, existing methods remove the individual effects before change point estimation through data transformations such as first-differencing. We prove that under reasonable assumptions, the unobserved individual effects bias has no impact on the consistent estimation of change points. Our simulations show that since our method does not remove any variation in the dataset before change point estimation, it performs better in small samples compared to first-differencing methods. We focus on short panels because they are commonly used in practice, and allow for the unobserved individual effects to vary over time. Our method is illustrated via two applications: the environmental Kuznets curve and the U.S. house price expectations after the financial crisis.

\end{abstract}

\section{Introduction}

In many panel datasets important variables of interest are missing, either because they are not available or because they are inherently unobservable. In a regression model of CO$_2$ emissions on energy consumption, variables such as a country's natural resources, political developments or the influence of environmental groups on decision making are typically not observed. While some of these unobserved variables or \textit{unobserved individual effects}, such as initial natural resources, may be time-constant, many others, like political developments or the influence of environmental groups, typically vary over time. Unobserved individual effects are common in panels such as cross-country data, survey data, medical studies, and when not properly dealt with, they cause slope estimates to be inconsistent since they introduce an omitted variable bias. To ensure consistency, most panel data methods assume that the individual effects are time-constant and remove them before estimation, through time-demeaning or first-differencing the initial model. In this paper, we show that the individual effects need not be removed for estimating the number and location of multiple change points. Despite the asymptotic bias in the slope estimators introduced by the unobserved individual effects, our method estimates consistently the number and location of change points under reasonable assumptions.

Besides the long line of work on change points in time series models (see \textit{interalia} Cs\"{o}rg\"{o} and Horv\'{a}th 1997; Bai and Perron 1998; Qu and Perron 2007; Harchaoui and L\`{e}vy-Leduc 2010; Aue and Horv\'{a}th 2013; Chan, Yau, and Zhang 2014; Perron and Yamamoto 2015; Qian and Su 2015) there is a growing body of work on change points in panel data. In panel data models, coefficients may exhibit common changes across customers, firms or countries due to for example policy changes, financial crises, housing bubbles or technological breakthroughs. Besides economics, panel change point methods have also been proposed to study changes in financial networks, in event counts for telecommunication networks, in encryptions of human speech and sound, and in genomic profiles of patients - see \textit{interalia} Vert and Bleakley (2010); Cho and Fryzlewicz (2015); Bardwell et al (2018).

Regarding testing for change points, several test statistics have been proposed by Emerson and Kao (2001, 2002), de Wachter and Tzavalis (2012), Horv\'{a}th and Hu\v{s}kov\'{a} (2012),  among others. Regarding estimation of multiple change points, one part of the literature is concerned with heterogeneous panels where the slope parameters are allowed to change across individuals or other panel units (see, e.g. Bai 2010; Kim 2011; Horv\'{a}th and Hu\v{s}kov\'{a} 2012;  Chan, Horv\'{a}th and Hu\v{s}kov\'{a} 2013; Torgovitski 2015; Cho and Fryzlewicz 2015; Cho 2016; Baltagi, Feng and Kao 2016; Okuy and Wang 2018). The other part, including this paper, focuses on homogenous panels, where the slope parameters are constant across individuals. In this setting, Vert and Bleakley (2010) propose estimating the change points via a group least-angle approach;  Qian and Su (2016) use an adaptive fused group Lasso (AGFL) method on the first-differenced data; Li, Qian and Su (2016) propose a principal component modified version of the AGFL method for dealing with a particular form of unobserved effects called interactive fixed effects;  Baltagi, Kao and Liu (2017) employ an OLS method on the initial as well as the first-differenced data for stationary and nonstationary regressors; Bardwell et al (2018) use a minimum description length criterion.

The majority of theoretical developments on change point estimation, including the ones listed above, focus on long panels (where either only the number of periods $T$ goes to infinity, or both $T$ and the number of individuals $N$ tend to infinity). However, many panels such as surveys, cross-country data and local administrative data remain short, either because the data has been discontinued, just started, or because older data is unreliable. Therefore, many applications using panel data are forced to rely on less than twenty time periods - see \textit{interalia} B\"orsch-Supan et al (2013) on the well-known European health survey SHARE, Baier and Bergstrand (2007) on a cross-country study of the impact of free trade agreements, Blanco and Ruiz (2013) on survey data analysis of the impact of crime on institutions and democracy,  Kyle and Williams (2016) on a cross country analysis of health care and prescription drugs, Niu and van Soest (2014) on the American Life Panel (ALP) survey and Armona, Fuster and Zafar (2018) on another self-collected survey of house price expectations. Only a few papers develop methods for change point estimation in short panels ($N \rightarrow \infty$, $T$ fixed) while at the same time considering unobserved individual effects. Bai (2010) and Torgovitski (2015) do so, but treat the unobserved individual effects as individual means with potential change points, and estimate these change points in the mean without including other regressors in the model; Bai (2010) proposes an OLS method and Torgovitski (2015) focuses on non-parametric estimation. Qian and Su (2016) consider long and short panels, and estimate a panel regression model with fixed effects, where they rely on first-differenced data for consistent estimation of multiple change points.

Our main contribution is to provide an OLS method that consistently estimates the number of change points in short panel regression models \textit{without transforming the initial data}. In our setting, the unobserved individual specific effects can either be treated as parameters or as random variables, and they may vary over time. Whether they are parameters or random variables, we omit them in the OLS estimation, and therefore \textit{a change point is defined as a change in the slope parameters, in the asymptotic bias of the OLS slope estimators, or in both.} After all these (pseudo) change points are identified, we consistently estimate the slope estimators via OLS estimation of the demeaned model in each corresponding stable sub-sample. This latter step allows us to test whether each of the identified change points can be attributed to changes in the slope parameters, often the quantities of interest to applied researchers.\footnote{The only changes that cannot be labeled as changes in slope parameters are those that occur exactly one period after another change, because in this case, with no further assumptions, any transformation would just remove the period in-between two adjacent changes and therefore any information about the corresponding slope parameters would be lost. This case is further discussed in Section 2.}

The literature often assumes that the unobserved individual effects are of a particular functional form, such as time-invariant individual effects (fixed effects), additive effects (fixed effects plus cross-section invariant time effects) or interactive fixed effects (fixed effects times a common shock that is cross-section invariant but changes over time, see e.g. Pesaran 2006; Bai 2009; Bai and Li 2014; Moon and Weidner 2015). Since in many applications, these assumptions can be perceived as too restrictive, we adopt a more general specification where individual specific effects can vary over time both in a smooth and abrupt way. 

The majority of papers (for long and short panels) that estimate changes in slope coefficients, like Qian and Su (2016), start from the premise that since OLS slope estimators are inconsistent due to the presence of unobserved individual effects, these effects need to be removed before change point estimation by means of some data transformation such as demeaning or first-differencing. In short panels, most of the variation is across individuals, and therefore such transformations, which typically remove a lot of cross-section variation, are problematic because they remove valuable information prior to change point estimation. Additionally, if the individual effects are not constant over the entire sample, first-differencing or any other available transformation does not fully remove them. In contrast to what most literature currently suggests, we prove that it is not necessary to transform the data for the purpose of change point estimation. Additionally, our simulation results show that in terms of correctly estimating the number of change points in small samples, our method performs better than the method in Qian and Su (2016), which relies on first-differencing.

Another contribution of this paper is to derive the asymptotic properties of two slope estimators while allowing for general time dependence and weak cross-section dependence in the level data. The first one is the conventional fixed effects estimator, obtained by OLS estimation of the initial model demeaned over each stable sub-sample, between two change points. For this estimator, we make the additional assumption that the unobserved individual effects change at the same time as the change in the slope parameters or the individual effects bias, which is still more general than assuming fixed, additive or interactive fixed effects. The second estimator is based on full-sample demeaning in the presence of fixed effects. We show that in the presence of fixed effects and change points, full-sample demeaning can lead to more efficient slope estimators as it uses the additional information that the individual effects do not change over time.

Related to this paper, for time-series models with regressors that are correlated with the errors, Perron and Yamamoto (2015) show under which conditions an OLS estimator for (pseudo) change points is consistent. They propose estimation of change points via sequential testing while our method consistently estimates the total number of (pseudo) change points in one step via an information criterion. Additionally, we show that if one imposes more change points than the truth (which may be desirable due to potential finite sample bias of post-selection methods), the set of estimated change points contains all the true ones with probability one in the limit.

The rest of the paper is organized as follows. Section 2 proves that our method consistently estimates the number and location of (pseudo) change points. Section 3 derives the asymptotic properties of the two proposed slope estimators. The finite-sample properties of the change point and slope estimators are studied through simulations in Section 4, and compared to the estimators in Qian and Su (2016). The practical use of our method is illustrated in Section 5 with two applications: the environmental Kuznets curve and the U.S. house price expectations in the aftermath of the financial crisis. In our first application, we show that in the implementation of the Kyoto protocol, major reductions in emission patterns occurred, which were unfortunately to a large extent undone after its implementation. In our second application, we show that determinants of house valuations changed from being largely subjective to being largely objective after the economy recovered from the recent financial crisis. Section 6 concludes. All the proofs are relegated to the Appendix.

\textbf{Notation: } Matrices and vectors are denoted with bold symbols, and scalars are not. Define for a scalar $S$, the generalized vec operator $\vc_{1:S}(\bm{A}_s) = (\bm A_1', \ldots, \bm A_S')'$, stacking in order the matrices $\bm A_s, (s=1,\ldots, S)$, which have the same number of columns. Let $\diag_{s=1:S}(\bm{A}_s) \equiv \diag_{1:S}(\bm{A}_s) = \diag(\bm A_1, \ldots, \bm A_{S})$ be the matrix that puts the submatrices $\bm A_1, \ldots, \bm A_S$ on the diagonal. If $S$ is the number of change points,  $T_1,\ldots, T_S$ are the ordered candidate change points and $T$ the number of time series observations, let $\lambda_0=0$, $\lambda_{S+1}=1$, and let $\vlam_S=(\lambda_0,\vc_{1:S}(\lambda_s)',\lambda_{S+1})'$ be \textbf{a sample partition} of the time interval $[1,T]$ divided by $T$, such that $\lambda_0=0$, $\lambda_{S+1}=1$, and $\lambda_s = T_s/T$ for $s=0,\ldots S+1$, with $T_0=0$ and $T_{S+1}=T$. Define constant regimes as $ I_s = [T_{s-1}+1, T_s]$ for $s=1, \ldots, S+1$. Let $\bm X = \vc_{1:T}(\bm X_t)$ be the $NT \times p$ matrix that stacks  $\bm X_t=\vc_{i=1:N}(\bm x_{it}')$ in order. Call $\tilX = \diag(\vc_{1:T_1}(\bm X_t), \vc_{T_1+1:T_2}(\bm X_t), \ldots, \vc_{T_{S}+1:T_{S+1}}(\bm X_t))$ \textbf{the diagonal partition} of $\bm X$ at $\bm \lambda_S$, with $\vX_1, \ldots, \vX_T$ on the diagonal and the rest of the elements zero. A superscript of $0$ on any quantity refers to the true quantity. For any random vector or matrix  $\bm Z$, denote by $||Z||$ the Euclidean norm for vectors, or the square root of the maximum eigenvalue of $\bm Z'\bm Z$ for matrices. Also denote $||\bm Z||_{q}$ the $\mathcal L_q$ norm, i.e. $||\bm Z||_{q}= E(||\bm Z ||^q)^{1/q}$. For convenience, we denote by $0$ either a scalar, a vector or a matrix of zeros, and we only specify its dimension when it is unclear.

\section{Change point estimation} \label{sec:model}

Assume that the true model is piecewise-linear with $m^0$ change points:
\begin{align} \label{eq:model}
y_{it}=
\bm x_{it}'\bm \beta_j^0+c_{it}+\epsilon_{it}, \qquad   t \in I_j^0, \qquad  j=1, \ldots, m^0+1.
\end{align}
In \eqref{eq:model}, $i=1,\ldots,N$ are individuals, $t=1,\ldots,T$ are time periods, with $N$ large  and  $T$  fixed, $y_{it}$ are scalar continuous outcomes, $\bm x_{it}$ is a $p \times 1$ vector including the intercept and observed covariates, some of which may be constant over time;  $m^0$ is the true  unknown number of change points, with $1 \leq m^0 \leq T-1$.  Also, $T_j^0,(j=1, \ldots, m^0)$ are the true unknown change points belonging to the sample partition $\vT_{m^0}^0$. \black{The true number and location of change points are properly defined in Assumption A\ref{a1}, and in model (1) they should be interpreted as possible changes in the unknown $p \times 1$ slope parameters $\vbeta_j^0$.} Furthermore, $\epsilon_{it}$ are unobserved mean-zero idiosyncratic errors, uncorrelated with $\vx_{it}$, and $c_{it}$ are the \textit{time-varying individual specific effects}, which are either parameters or \textit{unobserved random variables} that are uncorrelated with the idiosyncratic effects $\epsilon_{it}$ but possibly correlated with the observed covariates $x_{it}$. For example, in our second application, the subjective house price expectations equation contains unobservables related to individual optimism which may be correlated with covariates such as a home owner's view of his/her economic situation. \black{For the purpose of change point estimation, the time-variation allowed in $c_{it}$ is quite general and further discussed after Assumption A\ref{a1}.}

Assume first that the number of change points $m^0$ is known. To describe the least-squares change point estimators $\hat{\mathbf T}_{m^0} = \hat{\vlam}_{m^0} T$, let $u_{it} = c_{it} +\epsilon_{it}$, $\vu= \vc_{t=1:T}(\vc_{i=1:N} (u_{it}))$, $\vbeta^0=\vc_{j=1:m^0+1}(\vbeta_j^0)$, $\vy=\vc_{t=1:T}(\vc_{i=1:N} (y_{it}))$, and $\vXbar^0$ the diagonal partition of $\vX$ at the true partition $\vlam_{m^{0}}^0$. Then \eqref{eq:model} becomes:\begin{equation}\label{true}
\vy = \vXbar^0 \vbeta^0 + \vu.
\end{equation}

We propose estimating \eqref{true} by minimizing the sum of squared residuals over all possible sample partitions $\vlam_{m^0}$, which is equivalent to regressing $\vy$ on $\vXbar$ \black{(where the latter was defined in the notation section above),}
\begin{equation}\label{eq.s_nt}
\min_{\vlam_{m^0}} \, S_{NT}(\hat \vbeta_{\vlam_{m^0}}, \vlam_{m^0}) = \min_{\vlam_{m^0}}\, \left(NT\right)^{-1} \left(\vy-\vXbar \hat \vbeta_{\vlam_{m^0}} \right)'\left(y-\vXbar \hat \vbeta_{\vlam_{m^0}}\right),
\end{equation}
and where $\hat \vbeta_{\vlam_{m^0}} = (\vXbar' \vXbar)^{-1} \vXbar' \vy$ is the OLS estimator using $\vlam_{m^0}$ as the candidate partition. The minimizer of the above problem is denoted $\hat {\vlam}_{m^0}$ or $\hat {\mathbf T}_{m^0} =\hat {\vlam}_{m^0} T$ , and we refer to $\hat {\mathbf T}_{m^0}$ as the OLS change point estimators. If the minimizer is not unique, we break the tie by picking the smallest change point estimators. The OLS estimator of $\vbeta^0$ at the estimated partition is  denoted by $ \hat{\vbeta} = \hat{\vbeta}_{\hat \vlam_{m^0}}= \vc_{1:m^0+1}(\hat{\vbeta}_{j,\hat \vlam_{m^0}}) $. 

In general, $m^0$ is unknown and needs to be estimated. We propose estimating the number of change points by minimizing the following information criterion over $m=0,\ldots, T-1$, similar to BIC and HQIC:
$$
IC(m) = \log S_{NT}(\hat \vbeta_{\hat \vlam_{m}}, \hat \vlam_{m}) +p^*_m\ell_{NT},
$$
where $\ell_{NT}>0$, $\ell_{NT}\rightarrow 0$, $N \ell_{NT}\rightarrow \infty$, and $p_m^*$ is the number of parameters for a model with $m$ change points. Both Nimomiya (2005) - for a mean-shift model - and Hall and Sakkas (2013) - for a general regression model - show that the penalty for one change point should be three, not one, therefore we recommend to use $p_m^* = 3m+(m+1)p$. In the simulation section we show that the HQIC penalty, $\ell_{NT} = \log[\log (NT)]/NT$, is preferred to the BIC penalty, $\ell_{NT}= \log (NT)/NT$. The resulting estimator for the number of change points is $\hat m = \arg \min IC(m)$. Note that the information criteria is defined at the OLS change point estimators for a given number of change points, so we estimate the number and location of changes in one step.

For proving that our method \textit{consistently estimates the number and location of change points in $\vgam_j^0$, the pseudo-true parameters defined below}, we impose the following assumptions. 

\begin{assn}\label{a1} As $N\rightarrow \infty$: (i) $(NT)^{-1/2}\sum_{i=1}^N \sum_{t=1}^T \vx_{it} \epsilon_{it} \ind \mathcal N(0, V)$, where $V$ is a positive definite (pd) matrix of constants; (ii) $N^{-1} \sum_{i=1}^N \epsilon_{it} c_{it} \inp 0$ ;(iii)
        $ N^{-1} \sum_{i=1}^N \vx_{it} c_{it}\inp \va_j^0 $ for $t \in I_j^0$, $j=1,\ldots, m^0+1$;
        (iv) $ N^{-1} \sum_{i=1}^N \vx_{it}\vx_{it}'\rightarrow \vQ_j^0 $ for $t\in I_j^0$, $j=1,\ldots, m^0+1$, where $ \vQ_j^0$ are pd matrices of constants; (v) let $\vgam_j^0 = \vbeta_j^0+ (\vQ_j^0)^{-1} \va_j^0$; then $\vgam_j^0 \neq \vgam_{j+1}^0$ for all $j=1,\ldots,m^0$; (vi)$ (NT)^{-1} \sum_{i=1}^N \sum_{t=1}^T  \epsilon_{it}^2 \inp \sigma^2_{\epsilon,T}$ and $ (NT)^{-1} \sum_{i=1}^N \sum_{t=1}^T  c_{it}^2 \inp \sigma^2_{c,T}$. \end{assn}

A\ref{a1}(i) imposes a central limit theorem for \black{sums} of $\vx_{it} \epsilon_{it}$, allowing for general time-series dependence and for weak cross-section dependence. A\ref{a1}(ii) assumes that if the time-varying individual specific effects $c_{it}$ are random variables, they are uncorrelated with the idiosyncratic errors $\epsilon_{it}$, a common assumption in panels with individual specific effects. \black{If they are parameters, then they are also allowed to vary over time, but they are omitted in the estimation.}

 A\ref{a1}(iii)-(v) are key assumptions for consistent estimation of the number of change points. Note that they are not very restrictive in the sense that $\vgam_j^0$ can change at each point in time, and it may change because of $\vbeta_j^0$ or not. \black{The allowed time variation in $c_{it}$ is implicitly defined by A\ref{a1}(v), allowing $c_{it}$ to exhibit change points, smooth time-variation and/or jumps. The specification for $c_{it}$ includes fixed effects ($c_{it}=c_i$), interactive fixed effects ($c_{it}=c_i f_t$), but also (other) forms of stationary or non-stationary time variation.  \textit{Since we define the change points as changes in $\vgam_j^0$}, as long as $\vgam_j^0$ does not change, any time-variation in $c_{it}$ will not result in a change point. If $\vgam_j^0$ changes because of change points in $c_{it}$, then these change points are identified by our method.}

A\ref{a1}(vi) is a weak law of large numbers for sums of the second moments of $\epsilon_{it}$ and $c_{it}$, ensuring they do not increase with $N$. In Assumption A\ref{a2}, we consider a common set of primitive assumptions used for panel data (such as survey data), when the data is independent over $i$. Lemma \ref{lem1} below shows that A\ref{a2} satisfies A\ref{a1}(i)-(iv) and A\ref{a1}(vi).

\begin{assn}\label{a2}
(i) $\vx_{it}$ and $\epsilon_{it}$ are independent over $i$ with $E(\epsilon_{it})=0$ and $E(\epsilon_{it} \vx_{it}) =0$; (ii) $\sup_{it} ||\vx_{it}||_{4+\delta} <\infty$ and $\sup_{it} ||\epsilon_{it}||_{4+\delta} <\infty$ for some $\delta>0$.
(iii) $E(\vx_{it} \vx_{is}'\epsilon_{it} \epsilon_{is}) =\vV_{ts}^0$; (iv) $E(\vx_{it} \vx_{it}') =\vQ_j^0$ for $t\in I_j^0$; (v) $c_{it}$ are independent over $i$, with $E(c_{it} \epsilon_{it})=0$, $E(\vx_{it} c_{it}) = \va_j^0$ for $t \in I_j^0$ and $\sup_{it} ||c_{it}||_{4+\delta}<\infty$ ; (vi) $E(\epsilon_{it}^2)=\sigma_t^2$; (vii) $E(c_{it}^2)=\sigma_{ct}^2.$\end{assn}

\begin{lemma}\label{lem1}
If A\ref{a2} holds, then A\ref{a1}(i)-(iv) and (vi) holds.
\end{lemma}

\begin{theorem}\label{theo1}
Let A\ref{a1} hold. As $N\rightarrow \infty$,\\
(i) If $m=m^0$, $\lim P(\hat T_j = T_j^0)=1$  and $\hat \vbeta_{j, \hat \lambda_{m^0}} \inp \vgam_j^0$, for $j=1,\ldots, m^0+1$ \\
(ii) $\lim P(\hat m=m^0)=1$.\\
(iii) If $m>m^0$, then there are indices $j_1,\ldots j_{m^0} \in \{1,\ldots, m\}$ such that $\lim P[ \hat T_{j_s} = T_s^0] =1$ for all $s\in \{1,\ldots, m^0\}$.\\
\end{theorem}
Part (i) states that if we knew the true number of change points, their locations would be consistently estimated, and the corresponding parameter estimates would be consistent for their pseudo-true values $\vgam_j^0$. Part (ii) states that the true number of change points is \black{consistently estimated by the information criterion we propose.} Part (iii), a by-product of our proof, shows that if the number of changes imposed is larger than the true number of the changes, then our method estimates all the true change points with probability one in the limit (and some additional spurious change points).\footnote{This implies that an AIC variant of our information criterion, with penalty $\ell_{NT} =2/(NT)$, can also be used if the researcher is worried that the number of change points is underestimated.} 

The intuition for the result in Theorem \ref{theo1} is similar to Perron and Yamamoto (2015) who proposed using OLS for estimating change points in time series models with regressors that are correlated with errors. While the parameter estimates are in general not consistent for their true values because of the omitted variable bias, they are consistent for the pseudo-true values $\vgam_j^0$, therefore we can consistently estimate the number and location of change points in $\vgam_j^0$. Note that unlike Perron and Yamamoto (2015), we propose an  information criterion that consistently estimates the number of change points in one step. In contrast to Perron and Yamamoto (2015), we can also allow for a change point in each period.

The advantage of our method over Qian and Su (2016) and Li, Qian and Su (2016) is that we allow for time-varying individual effects without \black{specifying a functional form for the time variation.} Moreover, if some covariates are time-invariant as typical in panel data (gender, race), then the method in Qian and Su (2016), relying on first-differencing the data prior to change point estimation, only benefits from one period to estimate a change in the coefficients on these covariates, while our method uses all the periods available. Further advantages are highlighted in the simulation section, where we show that our method is more precise in estimating the number of change points \black{in finite samples}, because it does not remove important variation in the data by first-differencing. 

We now discuss the results in Theorem \ref{theo1} in connection to \black{assumption} A\ref{a1} and typical panel data assumptions. In its full generality outlined above, the method in this section does not yet indicate which changes occur only in the slope parameters $\vbeta_j^0$. However, because this may be of main interest to the applied researcher, the next section discusses under what conditions $\vbeta_j^0$ can be consistently estimated by two demeaning procedures. We then suggest using these estimators to test $H_0: \vbeta_j^0=\vbeta_{j+1}^0$, $j=1,\ldots, m^0$, therefore identifying which changes pertain to $\vbeta_j^0$ only.

\textit{In special cases, the changes in $\vbeta_j^0$ can be directly identified by the methods of this section.} The first case is when $c_{it}$ are random effects; in that case, they are uncorrelated with $\vx_{it}$, in which case a weak law of large numbers can be employed to show that $\va_j^0 = 0$, therefore that $\vgam_j^0=\vbeta_j^0$. The second case is if $\vQ_j^0=\vQ^0$ and $\va_j^0=\va^0$. In this case, the correlation between the individual specific effects and the regressors does not change over time, and therefore all the changes in $\vgam_j^0$ come from changes in $\vbeta_j^0$, and no further testing is needed.

Our method can also be used as a diagnostic tool for modeling either time-varying parameters or time-varying individual effects. If $\hat m$ is large (close or equal to $T-1$), then  \eqref{eq:model} should be revisited for better modelling of the time-variation in $\vgam_j^0$.  If $\hat m$ is small, then a model with interactive fixed effects, i.e. $c_{it} = c_i \times f_t$ might not be inappropriate unless it can be assumed that despite this specification, $\va_j^0$  does not change often. If the researcher is further willing to assume fixed effects (i.e. $c_{ij}=c_i$),  a common assumptions in panel data, then the next section provides a more efficient estimator of $\vbeta_j^0$ than currently available. 

If the covariates include a lagged dependent variable $y_{it-1}$, then $E(y_{it-1} c_{it})$ would in general change in each period, leading to $m^0=T-1$ by definition. Therefore, in our analysis, we do no include lagged dependent variables, but allow for time-series dynamics in the error term. \black{Employing time dummies for each time period, a common approach in short panels, is equivalent to imposing a change point in the intercept at each time period, which is not parsimonious nor necessarily justified by the data. We suggest avoiding this approach as our method can directly estimate the number and location of the changes in the intercept without having to assume they change at each period.}

\section{Slope Estimators and Their Asymptotic Properties}\label{sec:fe}
In this section, we proceed as if the true number and location of change points in $\vgam_j^0$ was known. For implementing the estimators described below, the true number and location of change points should be replaced by their corresponding estimates from Section 2.

To get consistent estimators of the slope parameters, we follow the common approach to remove the individual effects from model \eqref{eq:model} through a transformation of the data. \black{However, for proper removal of these effects, we assume throughout this section that $c_{it} = c_{ij}$ for $t\in I_j^0$, meaning that the individual effects are allowed to change but only at the change points already detected.\footnote{This assumption can be generalized to allow for further time variation in $c_{it}$. However, in this case, any transformations such as demeaning  would only remove the mean of $c_{it}$, leaving its variance in the error term. This variance would increase the variance of the slope estimators, but in general it would not be consistently estimable for fixed $T$ without further assumptions.}}

 We propose two methods: (1) sub-sample demeaning, that is, demeaning over segments $I_j^0$, corresponding to the usual fixed effects estimator in the absence of change points; and (2) full sample demeaning, which is new. The former is appropriate whether $\va_j^0$ is changing over $j$ or not, and the latter only when we have time-invariant fixed effects, i.e. $c_{it}=c_i$. 

 Note that only parameters that are constant for more than one period can be estimated via sub-sample demeaning, because the others are automatically removed. In contrast, when $c_{it}=c_i$, the full-sample demeaning estimator identifies all the parameters, including the ones for which only one time period is available. To our best knowledge, the full-sample demeaning estimator is new, and it is imposing the additional information that the fixed effects are not changing over time, which in principal should lead to more efficient estimators. In Theorem \ref{theo3}, we give sufficient conditions for this second estimator to be strictly more efficient than the first. Below, we describe these estimators.

For any vector $\vz_{it}$, let $ \overline \vz_i = T^{-1} \sum_{t=1}^T \vz_{it} $ be the full sample average, and $\overline \vz_{ij} = (T_j^0-T_{j-1}^0)^{-1} \sum_{t\in I_j^0} \vz_{it}$ be the sub-sample averages.  Then the FE estimator in $I_j^0$ is the OLS estimator in the demeaned sub-sample $I_j^0$ of model \eqref{eq:model}: $$\hat{\vbeta}_{FE,j}  = \left(\sum_{i=1}^N \sum_{t \in I_j^0} \ (\vx_{it} - \overline \vx_{ij})(\vx_{it} - \overline \vx_{ij})' \right)^{-1} \sum_{i=1}^N \sum_{t \in I_j^0}\ (\vx_{it} - \overline \vx_{ij})(y_{it}-\overline y_{ij}).$$

Let $\vw_{ij} =T^{-1}\sum_{t\in I_j^0} \vx_{it}$. If we demean model \eqref{eq:model} over the full sample, then
\begin{align*}
y_{it} -\overline y_i &=
(\vx_{it}-\vw_{ij})'\vbeta_j^0 - \sum_{s=1,s\neq j}^{m^0+1} \vw_{is}' \vbeta_s^0 + (\epsilon_{it}-\overline \epsilon_i), t\in I_j^0.
\end{align*}
Let $ y_{it}^* = y_{it} - \overline y_i $, $ \epsilon_{it}^* = \epsilon_{it} - \overline \epsilon_i $, and $ \widetilde \vx_{it} =(-\vw_1', -\vw_2',\ldots, -\vw_{j-1}',\vx_{it}' - \vw_j', - \vw_{j+1}',\ldots, -\vw_{m^0+1}')'.$ for $t\in I_j^0$. This model can be written more compactly as:
$
y_{it}^* = \widetilde \vx_{it}\vbeta^0 + \epsilon_{it}^*$, and the OLS estimator in this equation is the full sample demeaning estimator, which we name as the FFE (full sample fixed effects) estimator: 
$$\hat{\vbeta}_{FFE,j} = \left(\sum_{i=1}^N \sum_{t=1}^T \widetilde \vx_{it} \widetilde \vx_{it}'\right)^{-1}
\left(\sum_{i=1}^N \sum_{t=1}^T \widetilde \vx_{it} y_{it}^*\right).
$$
Let $\mathcal S$ be the subset of the number of regimes with at least two observations, and denote a quantity defined over $\mathcal S$ by a subscript $\mathcal S$; for example, $\vbeta_{\mathcal S}^0 = \vc_{j\in \mathcal S}(\vbeta_j^0)$.

\begin{assn}\label{a3} As $N \rightarrow \infty$, (i) $  N^{-1} \sum_{i=1}^N \vx_{it}\vx_{is}' \inp \vOmega_{ts}^0 $, a matrix of constants, for all $ t,s $; (ii) $  N^{-1/2}  \sum_{i=1}^N  \vc_{j \in \mathcal S}\left(\sum_{t\in I_j^0} (\vx_{it}-\overline \vx_{i,j})\epsilon_{it}\right) \ind \mathcal N(0,\vW_{1,\mathcal S}) $; (iii) $ N^{-1/2}  \sum_{i=1}^N \left(\sum_{t=1}^T \widetilde \vx_{it} \epsilon_{it}^* \right) \ind \mathcal N(0, \vW_{2})$.\footnote{The quantities $\vW_1$ and $\vW_2$ may depend on $T$, but for simplicity we do not explicitly express this in the notation.}
\end{assn}

Assumption A\ref{a3} facilitates the presentation of the asymptotic distributions for general time series dependence (including unit root dependence over $t$ in $\epsilon_{it}$) and weak cross-section dependence. Assumptions \ref{a4} gives a primitive assumption that satisfies Assumption A\ref{a3}. Let $\vX_i=\vc_{t=1:T}(\vx_{it})$. Then:

\begin{assn}\label{a4}
A\ref{a1}(i)-(iv), (vi) holds and: (i) $E(\epsilon_{it}\vX_i) = 0$ for all $i,t$; (ii) $E(\vx_{it} \vx_{is}') =\vOmega_{ts}^0$; (iii) $E(\vx_{it}\vx_{is}' \epsilon_{it} \epsilon_{is}) = \bm N_{ts}^0$.\end{assn}

\begin{lemma}\label{l1}
If A\ref{a4} is satisfied, then A\ref{a3} are satisfied. 
\end{lemma}

Let $\Delta T_j^0 = T_j^0-T_{j-1}^0$, $\vOmega_{1,\mathcal S} =\diag_{j\in \mathcal S}\left( \Delta T_j^0 \vQ_j^0 - (\Delta T_j^0)^{-1} \vQ_{jj}^0 \right)$, $\vQ_{jj}^0=\sum_{t\in I_j^0}\sum_{s\in I_k^0} \vOmega_{ts}^0$, and let $\vOmega_2= \diag_{j=1:m^0+1}(\Delta T_j^0 \vQ_j^0) - (2T^{-1} -T^{-2}) \widetilde \vQ^0,$
where $\widetilde \vQ^0$ is the $p(m^0+1) \times p(m^0+1)$ matrix with the $(j,k)$ sub-matrix of size $p\times p$ equal to $\vQ_{jk}^0$.

\begin{theorem}\label{theo2}
Under A\ref{a1}, A\ref{a3} and $N\rightarrow \infty$,\\
 (i) 
$
\sqrt{N} \left( \hat{\vbeta}_{FE,\mathcal S} - \vbeta_{\mathcal S}^0\right) \stackrel{d}{\rightarrow} \mathcal{N}(0,\bm V_{FE,\mathcal S}),
$
where $\bm V_{FE,\mathcal S} = \vOmega_{1,\mathcal S}^{-1} \vW_{1,\mathcal S}\vOmega_{1,\mathcal S}^{-1} $; (ii) further assuming that $c_{it}=c_i$ for all $t=1,\ldots, T$, 
$
\sqrt{N} (\hat{\vbeta}_{FFE} -\vbeta^0) \ind \mathcal N\left (0, \bm V_{FFE}\right), $
where $\bm V_{FFE} = \vOmega_2^{-1}\vW_{2} \vOmega_2^{-1}$.
\end{theorem}

\black{It is interesting to note that even though the demeaning removes time-invariant regressors such as gender and race for both estimators, the magnitude of change in the slopes of these regressors can be consistently estimated.\footnote{The elements of $\bm {\hat \beta}_{FE,\mathcal S}$ and $\bm \hat \beta_{FFE}$ in Theorem \ref{theo2} referring to time-invariant regressors should be replaced by magnitudes of changes, but we did not do this to simplify notation.}}

With no further assumptions on the time series dependence, it is unclear which estimator is more efficient; therefore, when it can be assumed that $c_{it}=c_i$, we suggest stacking the moment conditions implied by these two estimators, resulting in a generalized method of moments estimator which is more efficient than each of the two if the optimal weighting matrix is used.

Theorem \ref{theo2} shows that when the data is uncorrelated over time (as in typical static panels), the FFE estimator is strictly more efficient than the FE estimator, so if it can safely be assumed that $c_{it}=c_i$, then the FFE estimator is preferable. Since this result pertains only to panel data models with at least one change point and with two parameters that can be identified by FFE, we impose $T\geq 4$, $m^0\geq 1$ and $\Delta T_j^0 \geq 2$ for at least two regimes.

\begin{theorem}\label{theo3}
Let A\ref{a4} hold, $m^0>1$, $T\geq 4$ and $\Delta T_j^0 \geq 2$ for at least two regimes $j \in\{1,\ldots, m^0+1\}$. Assume that $E(\epsilon_{it}^2|\vX_i) = \sigma^2$ for all $t$, and $E(\epsilon_{it}\epsilon_{is}|\vX_i)=0$ and $E(\vx_{it} \vx_{is})=\vOmega_{ts}^0=0$ for all $t\neq s$. Then $$\bm V_{FE} = \diag_{1:m^0+1}\left[\sigma^2 (\vQ_j^0)^{-1} \frac{1}{\Delta T_j^0-1}\right], 
\bm V_{FFE} = \diag_{1:m^0+1}\left[\sigma^2 (\vQ_j^0)^{-1}\frac{(T^2-3T+1) T^2}{(T-1)^4} \frac{1}{\Delta T_j^0}\right],$$ and $\bm V_{FFE}- \bm V_{FE}$ is negative definite.

\end{theorem}

Theorem \ref{theo3} shows that the relative efficiency of the FFE estimator can be explicitly quantified.  

\textbf{Overall modelling strategy. }Since this section provides consistent estimators in each period with more than one observation, one can test which parameters are actually changing by testing $H_0: \vbeta_j^0 = \vbeta_{j+1}^0$ for $j \in \mathcal S$, for example, via a Wald test at the level $\alpha$. Using a simple Bonferroni correction to correct for multiple testing, the overall size of the testing procedure is no larger than $\alpha \hat m$. Also note that $H_0: \vQ_j^0 = \vQ_{j+1}^0$ for $j=1,\ldots, m^0+1$ is testable via a Wald test. So if no changes in $\vbeta_j^0, \vQ_j^0$ are detected for two adjacent regimes, it means that all changes are coming from the individual effects, offering evidence for time-varying individual effects. Similarly, one can identify which change points pertain to the intercept alone by testing only the first restriction in $H_0: \vbeta_j^0 = \vbeta_{j+1}^0$, informing the researcher which periods actually need time dummies. \black{If one is worried about post-model selection issues after the estimation of the number of change points, then one should impose more change points than found by the HQIC.}

\section{Simulation Study}
\label{sec:sim}

This section looks at the finite sample performance of $\hat{\mathbf T}_{\hat m}$, $ \hat{\vbeta}_{OLS} \equiv \hat \vbeta$, $ \hat{\vbeta}_{FE} $ and $ \hat{\vbeta}_{FFE} $. The data generating process (DGP) is based on model (\ref{eq:model}) \black{with fixed effects:} $c_{it} = c_i$. The idiosyncratic errors $\epsilon_{it}$ and $c_i$ are independently drawn from $N(0,1/4)$. A single regressor is generated as $x_{it}=\sqrt{2}c_i+z_{it}$, with $z_{it}\sim iid~N(0,1/2)$. The vector of slope parameters $\vbeta^0 $ has elements alternating between $-0.1$ and $0.1$ for the different regimes between change points. For the case of one change point, we consider $T^0_1=2$ and $T^0_1=[T/3]$, where $[\cdot]$ is the least integer function; \black{for two change points, we let $T_1^0=[T/3]$ and $T_2^0=[2T/3]$. We let  $N=50,100,500$ and $T=20,30,50$. All results are reported for $1000$ replications.}

\begin{figure}[h!]
\caption{Histogram of estimated locations for a single change point}
\label{fig:case1}
\vspace{3mm} 
\begin{minipage}[b]{0.49\linewidth}
\centering
\scriptsize{$T=20$, change point at $[T/3]$}\\ \vspace{-0.3mm} 
\includegraphics[scale=0.6]{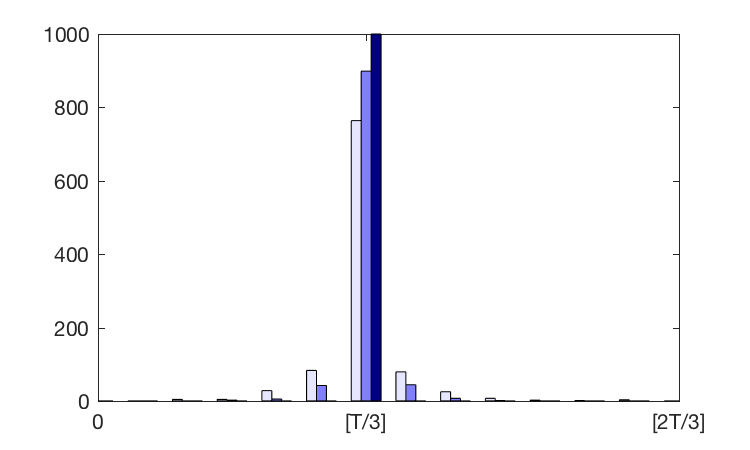}\\
\scriptsize{$T=30$, change point at $[T/3]$}\\ \vspace{-0.3mm} 
\includegraphics[scale=0.6]{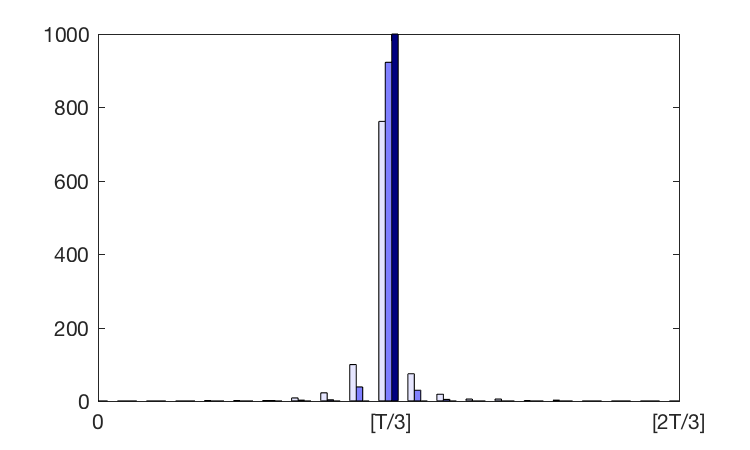}\\
\scriptsize{$T=50$, change point at $[T/3]$}\\ \vspace{-0.3mm} 
\includegraphics[scale=0.6]{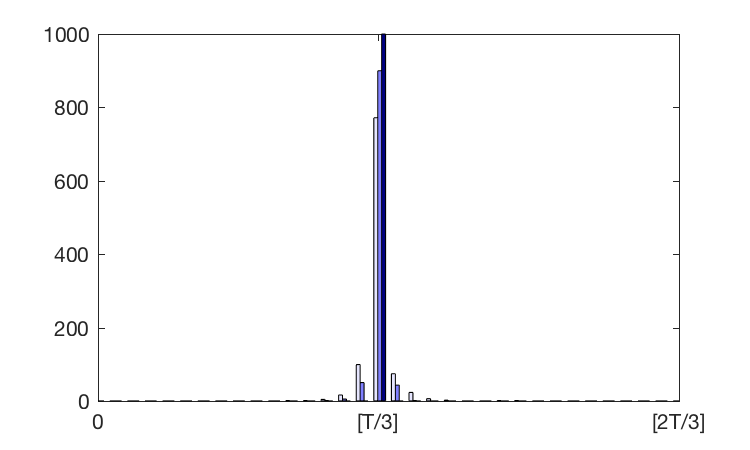}\\
\end{minipage}
\begin{minipage}[b]{0.49\linewidth}
\centering
\scriptsize{$T=20$, change point at $T=2$}\\ \vspace{-0.3mm} 
\includegraphics[scale=0.6]{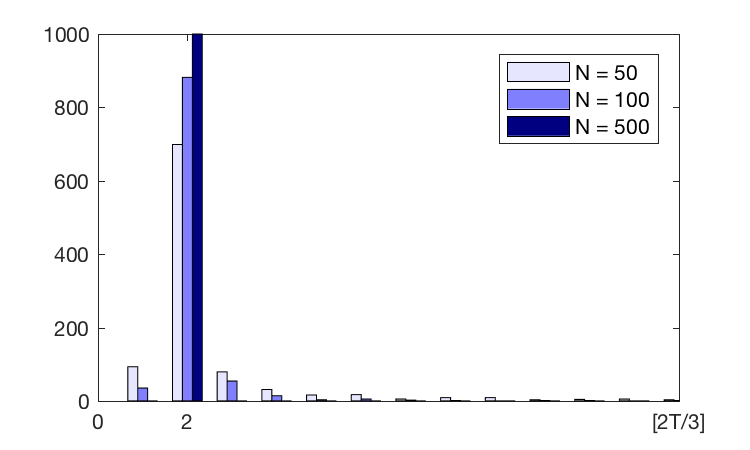}\\
\scriptsize{$T=30$, change point at $T=2$}\\ \vspace{-0.3mm} 
\includegraphics[scale=0.6]{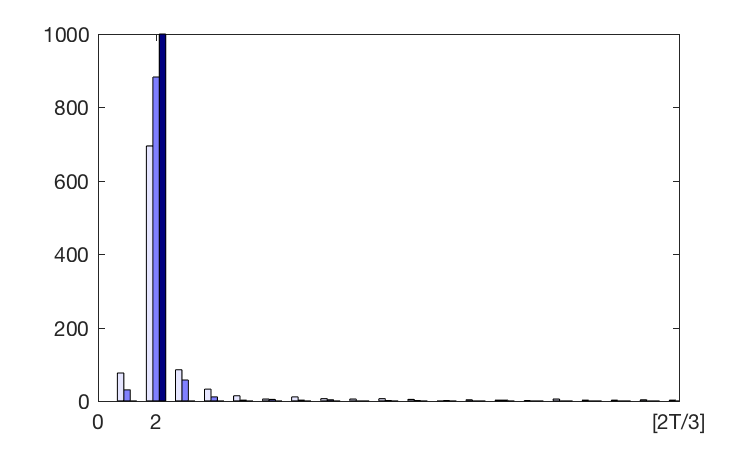}\\
\scriptsize{$T=50$, change point at $T=2$}\\ \vspace{-0.3mm} 
\includegraphics[scale=0.6]{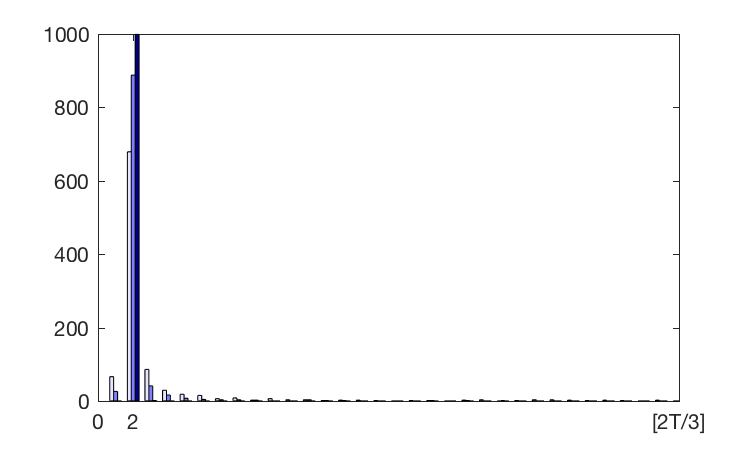}\\
\end{minipage}
\end{figure}

\begin{figure}[h]
\caption{Histogram of estimated locations for two change points}
\label{fig:case2}
\centering
\includegraphics[scale=0.6]{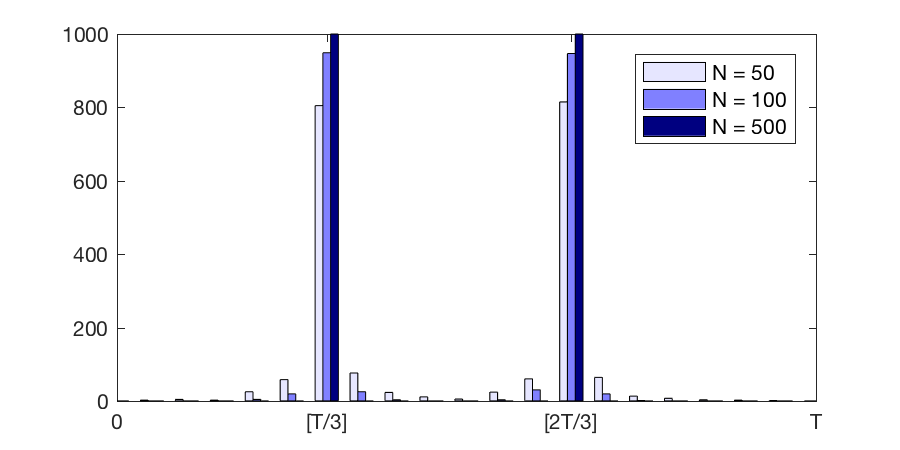}\\
\scriptsize{$T=20$; change points at [T/3] and [2T/3]}\\
\end{figure}

\black{Figures \ref{fig:case1}-\ref{fig:case2} report histograms of the estimated change point locations, assuming that $m^0$ is known. Similarly, Tables \ref{tab:case1}-\ref{tab:case2} report slope estimators based on the true sample partition of the respective DGP. Since for $T^0_{1}=[T/3]$ the change point location increases proportionally with $T\in\{20,30,50\}$, the number of time periods before and after the change point are more balanced compared to the case with $T^0_{1}=2$. As a consequence, in \black{the left panels of Figure \ref{fig:case1}}, the distribution of estimates is centered at the true change point $[T/3]$, while \black{the right panels show a distribution that is skewed to the right for small $N$}. As $N$ increases, the distribution collapses over the true change point for both choices of $T^0_{1}$. For two change points, Figure \ref{fig:case2} shows that the estimated locations are also increasingly accurate as $N$ grows, as is expected from the consistency result of Theorem \ref{theo1}.}

\FloatBarrier
\black{For the same DGPs, Tables \ref{tab:case1} and \ref{tab:case2} report bias, standard error and mean squared error (MSE) of slope estimates based on OLS, FE and FFE estimates averaged over 1000 repetitions. Standard errors are calculated based on Theorem \ref{theo3}.\footnote{\black{Note that standard errors of the OLS estimators cannot be computed due to not observing the individual effects.}} Due to the dependence of $x_{it}$ on $c_i$, OLS estimators exhibit, as expected, a strong bias. Overall, FE and FFE estimators perform well even in small samples ($N=50$; $T=20$) with average bias close to zero and small standard errors. For both choices of change point locations, FFE estimators have smaller standard errors compared to FE, as expected from Theorem \ref{theo3}.}

\begin{table}[htp!]
        \scriptsize
        \caption{Properties of slope estimators for a single change point at $[T/3] $}\label{tab:case1}
        \centering
        \begin{tabular}{|c|c|c|ccc|ccc|}                                                                                                 
                \hline                                                  &                                               &                         & \multicolumn{3}{c}{$ \hat \beta_1 $} \vline & \multicolumn{3}{c}{$ \hat \beta_2 $} \vline  \\
                &                                               &                       & Bias & SE & MSE & Bias & SE & MSE  \\ \hline
                \multirow{9}{*}{$N=50$} & \multirow{3}{*}{$T=20$} & OLS         & 0.352 & -                       & -             & 0.352 & -                     & -       \\                                                                                                                                                                                                      
                \cline{3-9}                                     &                                               & FE              & 0.001 & 0.045         & 0.002         & 0.001 & 0.028         & 0.001  \\ 
                \cline{3-9}                                     &                                               & FFE     & 0.000 & 0.042         & 0.002         & 0.001 & 0.027         & 0.001    \\  
                \cline{2-9}                                     & \multirow{3}{*}{$T=30$} & OLS   & 0.349         & -             & -             & 0.349         & -               & -              \\                                                    
                \cline{3-9}                                     &                                               & FE              & 0.000         & 0.033 & 0.001         & -0.001 & 0.023 & 0.001          \\   
                \cline{3-9}                                     &                                               & FFE     & -0.001        & 0.032 & 0.001         & -0.001 & 0.023 & 0.001          \\ 
                \cline{2-9}                                     & \multirow{3}{*}{$T=50$} & OLS   & 0.349         & -             & -             & 0.349 & -                     & -                \\                                                                                                                                                                                                      
                \cline{3-9}                                     &                                               & FE              & 0.000         & 0.026 & 0.001 & -0.000 & 0.017        & 0.000  \\                                                                                                                                                     
                \cline{3-9}                                     &                                               & FFE     & -0.001 & 0.025 & 0.001 & -0.000 & 0.017       & 0.000  \\ \hline
                \multirow{9}{*}{$N=100$}        & \multirow{3}{*}{$T=20$} & OLS   & 0.352  & -            & -             & 0.352         & -             & -       \\
                \cline{3-9}                                     &                                               & FE              & -0.000 & 0.032        & 0.001         & -0.000 & 0.020         & 0.000         \\
                \cline{3-9}                                     &                                               & FFE     & -0.000 & 0.030        & 0.001         & 0.000         & 0.019         & 0.000   \\ 
                \cline{2-9}                                     & \multirow{3}{*}{$T=30$} & OLS   & 0.351 & -                     & -             & 0.352 & -                     & -               \\ 
                \cline{3-9}                                     &                                               & FE              & 0.000 & 0.024         & 0.001 & 0.000 & 0.016         & 0.000 \\
                \cline{3-9}                                     &                                               & FFE     & -0.000 & 0.023 & 0.001 & 0.000 & 0.016        & 0.000 \\ 
                \cline{2-9}                                     & \multirow{3}{*}{$T=50$} & OLS   & 0.352 & -                     & -             & 0.352 & -                     & -               \\
                \cline{3-9}                                     &                                               & FE              & -0.001 & 0.018 & 0.000        & 0.000 & 0.012         & 0.000 \\
                \cline{3-9}                                     &                                               & FFE     & -0.000 & 0.018 & 0.000        & 0.000 & 0.012         & 0.000 \\ \hline
                \multirow{9}{*}{$N=500$} & \multirow{3}{*}{$T=20$} & OLS         & 0.353         & -             & -             & 0.353 & -                     & -       \\
                \cline{3-9}                                     &                                               & FE              & -0.000 & 0.014        & 0.000         & 0.000 & 0.009         & 0.000 \\
                \cline{3-9}                                     &                                               & FFE     & 0.000         & 0.013         & 0.000         & 0.000 & 0.009         & 0.000 \\ 
                \cline{2-9}                                     & \multirow{3}{*}{$T=30$} & OLS   & 0.353 & -                     & -             & 0.353 & -                     & -                       \\
                \cline{3-9}                                     &                                               & FE              & -0.000 & 0.011        & 0.000 & 0.000 & 0.007         & 0.000 \\
                \cline{3-9}                                     &                                               & FFE     & 0.000 & 0.010         & 0.000 & 0.000 & 0.007         & 0.000 \\ 
                \cline{2-9}                                     & \multirow{3}{*}{$T=50$} & OLS   & 0.353 & -                     & -             & 0.353 & -                     & -                       \\
                \cline{3-9}                                     &                                               & FE              & 0.000 & 0.008         & 0.000 & 0.000 & 0.006         & 0.000 \\
                \cline{3-9}                                     &                                               & FFE     & 0.000 & 0.008         & 0.000 & 0.000 & 0.005         & 0.000 \\ \hline
        \end{tabular}
\end{table}
\begin{table}[htp!]
        \scriptsize
        \caption{Properties of slope estimators for one change point at $2 $}\label{tab:case2}
        \centering
        \begin{tabular}{|c|c|c|ccc|ccc|}                                                                                                 
                \hline                                                  &                                               &                         & \multicolumn{3}{c}{$ \hat \beta_1 $} \vline & \multicolumn{3}{c}{$ \hat \beta_2 $} \vline  \\
                &                                               &                       & Bias & SE & MSE & Bias & SE & MSE  \\ \hline
                \multirow{9}{*}{$N=50$} & \multirow{3}{*}{$T=20$} & OLS         & 0.352 & -                       & -             & 0.352 & -             & -               \\                                                                                                                                                                                                    
                \cline{3-9}                                     &                                                                               & FE              & 0.002 & 0.100                 & 0.010 & 0.001 & 0.024         & 0.001 \\                                                                                                                                                   
                \cline{3-9}                                     &                                                                               & FFE     & -0.001 & 0.073        & 0.005 & 0.001 & 0.024 & 0.001          \\
                \cline{2-9}                                     & \multirow{3}{*}{$T=30$} & OLS           & 0.352 & -                     & -             & 0.349 & -               & -             \\                                                                                                                                                                                                                                              
                \cline{3-9}                                     &                                                                               & FE              & 0.001 & 0.100                 & 0.010 & -0.001 & 0.019 & 0.000 \\                                                                                                                                                    
                \cline{3-9}                                     &                                                                               & FFE     & -0.000 & 0.072        & 0.005 & -0.001 & 0.019 & 0.000 \\
                \cline{2-9}                                     & \multirow{3}{*}{$T=50$} & OLS           & 0.349 & -                     & -             & 0.349 & -               & -             \\                                                                                                                                                                                                        
                \cline{3-9}                                     &                                                                               & FE              & 0.003 & 0.100                 & 0.010 & -0.000 & 0.015 & 0.000  \\                                                                                                                                                                                                                                                                                          
                \cline{3-9}                                     &                                                                               & FFE     & -0.000 & 0.071        & 0.005 & -0.000 & 0.015 & 0.000 \\ \hline
                \multirow{9}{*}{$N=100$}        & \multirow{3}{*}{$T=20$} & OLS   &  0.350 & -            & -             & 0.352 & -             & -               \\
                \cline{3-9}                                     &                                                                               & FE              &  -0.000       & 0.071 & 0.005         & 0.000 & 0.017         & 0.000 \\
                \cline{3-9}                                     &                                                                               & FFE     &  -0.003       & 0.051 & 0.003 & 0.000 & 0.017 & 0.000  \\
                \cline{2-9}                                     & \multirow{3}{*}{$T=30$} & OLS           &       0.351   & -             & -             & 0.352 & -               & -             \\ 
                \cline{3-9}                                     &                                                                               & FE              &       -0.001  & 0.071 & 0.005 & 0.000 & 0.014 & 0.000  \\
                \cline{3-9}                                     &                                                                               & FFE     &       -0.000  & 0.051 & 0.003         & 0.000 & 0.014 & 0.000  \\ 
                \cline{2-9}                                     & \multirow{3}{*}{$T=50$} & OLS           &       0.354   & -             & -             & 0.352 & -               & -             \\
                \cline{3-9}                                     &                                                                               & FE              &       0.003   & 0.071 & 0.005 & -0.000 & 0.010 & 0.000  \\
                \cline{3-9}                                     &                                                                               & FFE     &  0.002        & 0.051 & 0.003 & -0.000 & 0.010 & 0.000  \\ \hline
                \multirow{9}{*}{$N=500$} & \multirow{3}{*}{$T=20$} & OLS &  0.353 & -            & -             & 0.353 & -             & -             \\
                \cline{3-9}                                     &                                                                               & FE              &  -0.001 & 0.032 & 0.001 & 0.000 & 0.008 & 0.000  \\
                \cline{3-9}                                     &                                                                               & FFE     &       0.000 & 0.023   & 0.001 & 0.000 & 0.008 & 0.000 \\  
                \cline{2-9}                                     & \multirow{3}{*}{$T=30$} & OLS           &  0.353 & -            & -             & 0.353 & -             & -               \\
                \cline{3-9}                                     &                                                                               & FE              &  -0.001 & 0.032 & 0.001 & 0.000 & 0.006 & 0.000 \\
                \cline{3-9}                                     &                                                                               & FFE     &       0.000 & 0.023 & 0.001 & 0.000 & 0.006 & 0.000 \\ 
                \cline{2-9}                                     & \multirow{3}{*}{$T=50$} & OLS           &       0.352 & -               & -             & 0.353 & -               & -     \\
                \cline{3-9}                                     &                                                                               & FE              &       -0.002 & 0.032 & 0.001 & 0.000 & 0.005 & 0.000 \\
                \cline{3-9}                                     &                                                                               & FFE     &       -0.000 & 0.023 & 0.001 & 0.000 & 0.005 & 0.000 \\ \hline
        \end{tabular}
\end{table}

\FloatBarrier

\begin{figure}[htp!]
\caption{Estimated number of change points for $T=20$}
\label{fig:case3}
\vspace{3mm} 
\normalsize{HQIC} \\ 
\begin{minipage}[b]{0.33\linewidth}
\centering
\vspace{2mm} \scriptsize{No change point}\\ \vspace{-0.3mm} 
\includegraphics[scale=0.5]{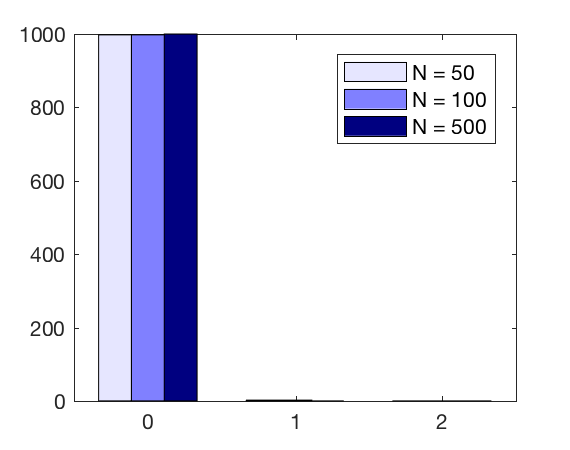}\\
\end{minipage}
\begin{minipage}[b]{0.33\linewidth}
\centering
\scriptsize{One change point at [T/3]}\\ \vspace{-0.3mm} 
\includegraphics[scale=0.5]{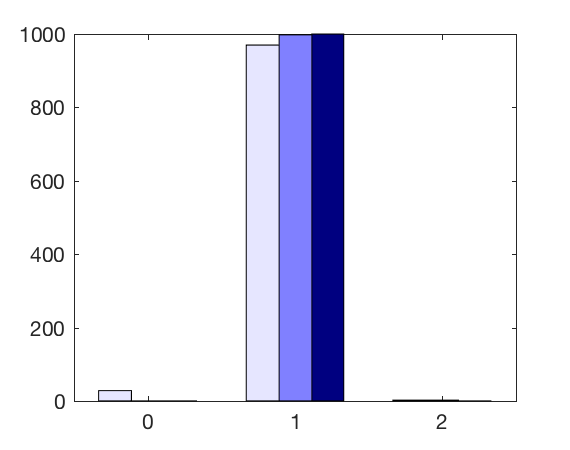}\\
\end{minipage}
\begin{minipage}[b]{0.33\linewidth}
\centering
\scriptsize{Two change points at [T/3] and [2T/3]}\\ \vspace{-0.3mm} 
\includegraphics[scale=0.5]{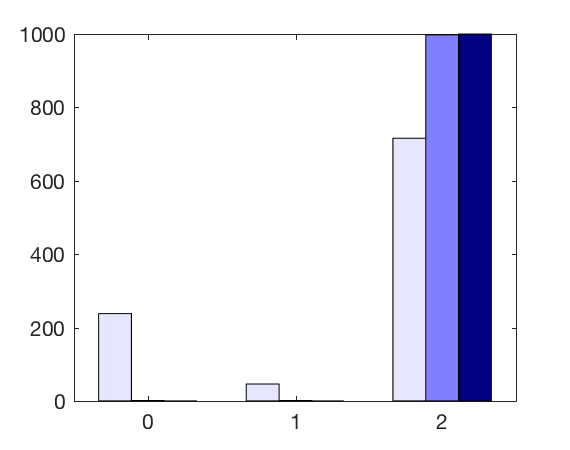}\\
\end{minipage}

\vspace{3mm} 
\normalsize{BIC} \\ 
\begin{minipage}[b]{0.33\linewidth}
\centering
\vspace{2mm} \scriptsize{No change point}\\ \vspace{-0.3mm} 
\includegraphics[scale=0.5]{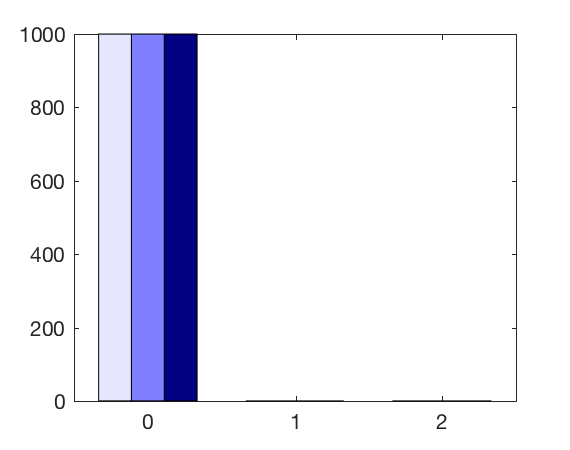}\\
\end{minipage}
\begin{minipage}[b]{0.33\linewidth}
\centering
\scriptsize{One change point at [T/3]}\\ \vspace{-0.3mm} 
\includegraphics[scale=0.5]{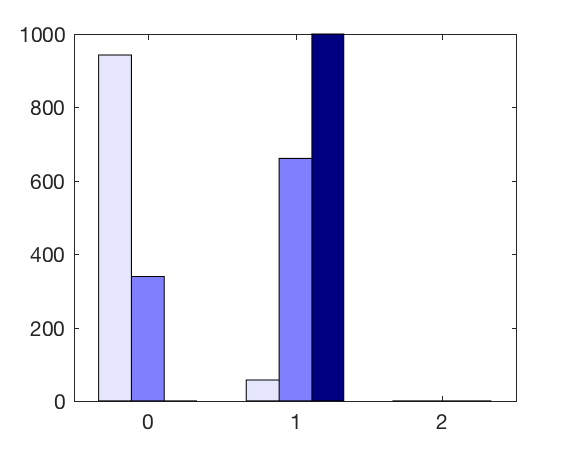}\\
\end{minipage}
\begin{minipage}[b]{0.33\linewidth}
\centering
\scriptsize{Two change points at [T/3] and [2T/3]}\\ \vspace{-0.3mm} 
\includegraphics[scale=0.5]{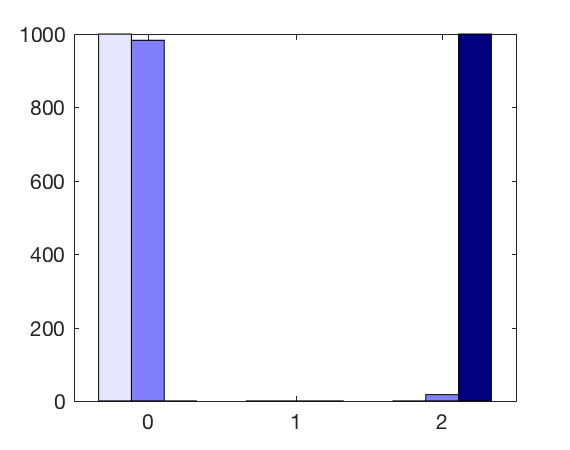}\\
\end{minipage}
\end{figure} 

Figure \ref{fig:case3} shows the estimated number of change points for the HQIC and BIC criteria defined in Section 2, and DGPs with $0,1$ or $2$ change points. \black{ Both HQIC and BIC perform well in large samples, but in smaller samples, the BIC strongly underestimates the number of change points}. For this reason, we use HQIC in both applications of Section \ref{sec:app}. 

\begin{figure}[htp!]
        \caption{Estimated number of change points for $T = 20$; 2 change points (at [T/3] and [2T/3])}\label{fig:case4}
        \vspace{3mm} 
        \normalsize{HQIC} \\ 
        \begin{minipage}[b]{0.33\linewidth}
                \centering
                \vspace{2mm} \scriptsize{$w = 0~$}\\ \vspace{-0.3mm} 
                \includegraphics[scale=0.6]{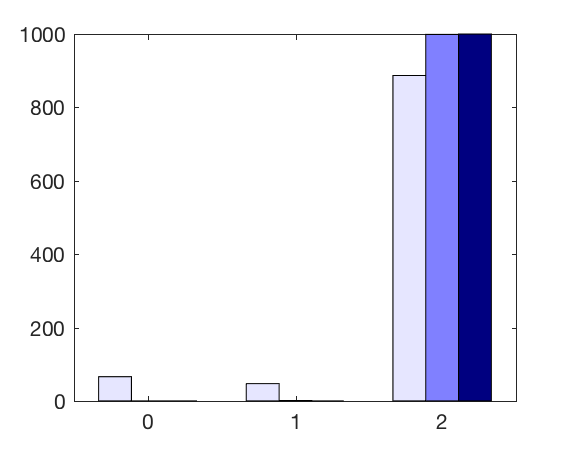}\\
        \end{minipage}
        \begin{minipage}[b]{0.33\linewidth}
                \centering
                \vspace{2mm} \scriptsize{$w = 0.1$}\\ \vspace{-0.3mm} 
                \includegraphics[scale=0.6]{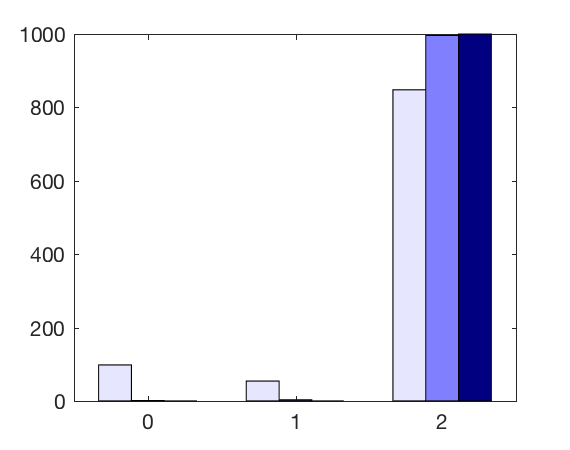}\\
        \end{minipage}
        \begin{minipage}[b]{0.33\linewidth}
                \centering
                \vspace{2mm} \scriptsize{$w = 0.3$}\\ \vspace{-0.3mm} 
                \includegraphics[scale=0.6]{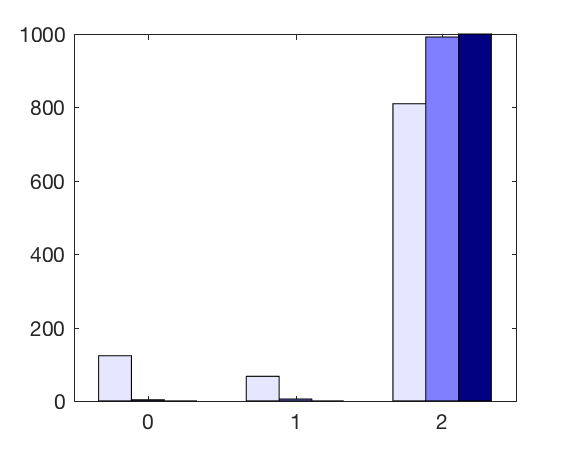}\\
        \end{minipage} \vspace{-1mm} \\ 
        \normalsize{AGFL (Qian and Su, 2016)} \\ 
        \begin{minipage}[b]{0.33\linewidth}
                \centering
                \vspace{2mm} \scriptsize{$w = 0~$}\\ \vspace{-0.3mm} 
                \includegraphics[scale=0.6]{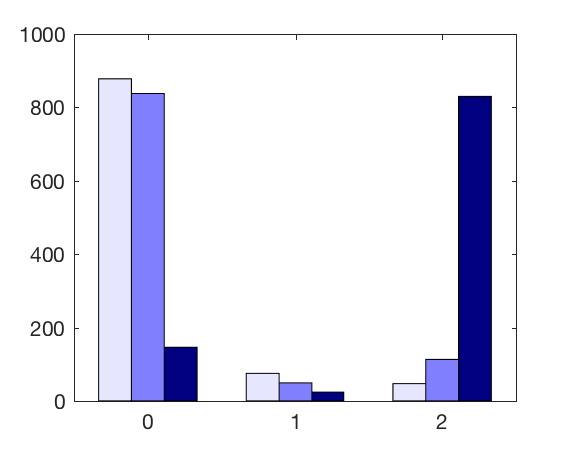}\\
        \end{minipage}
        \begin{minipage}[b]{0.33\linewidth}
                \centering
                \vspace{2mm} \scriptsize{$w = 0.1$}\\ \vspace{-0.3mm} 
                \includegraphics[scale=0.6]{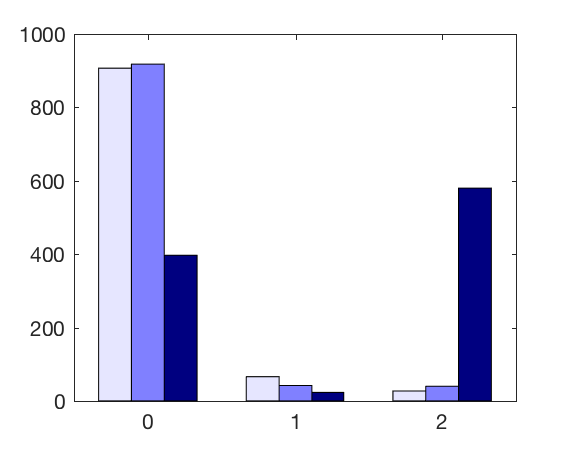}\\
        \end{minipage}
        \begin{minipage}[b]{0.33\linewidth}
                \centering
                \vspace{2mm} \scriptsize{$w = 0.3$}\\ \vspace{-0.3mm} 
                \includegraphics[scale=0.6]{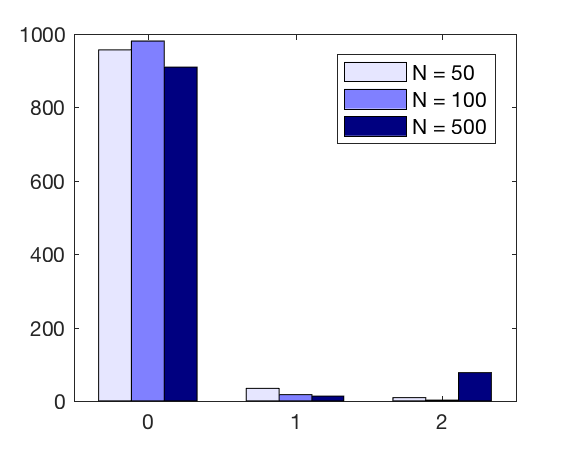}\\
        \end{minipage}
\end{figure} 

\begin{table}[htp!]                                                                                                                     
        \centering       
        \small 
        \resizebox{\columnwidth}{!}{                                                                                                                    
                \begin{tabular}{|c|ccc|ccc|ccc|}    
                        \hline &  \multicolumn{3}{c}{ AGFL } \vline &   \multicolumn{3}{c}{ FE } \vline   &   \multicolumn{3}{c}{ FFE } \vline   \\                                                                                                                                  
                        & $\hat{\beta_1}$ & $\hat{\beta_2}$ & $\hat{\beta_3}$ & $\hat{\beta_1}$ & $\hat{\beta_2}$ & $\hat{\beta_3}$ & $\hat{\beta_1}$ & $\hat{\beta_2}$ & $\hat{\beta_3}$ \\                                   
                        \hline                                                                                                                                
                        $w=0$           & -0.103        & 0.103          & -0.101       & -0.101                 & 0.100         & -0.100                &  -0.101       & 0.100   & -0.100 \\                                                                                                                                             
                        & (0.124) & (0.116) & (0.119)   & (0.013)       & (0.012) & (0.012)       &  (0.011)      & (0.010) & (0.010) \\  
                        \hline                                                                                                                           
                        $w=0.1$  & -0.106       & 0.108                 & -0.103                & -0.101          & 0.100                 & -0.100                & -0.101                 & 0.100         & -0.100 \\                                                                                                                                          
                        & (0.130)       & (0.121)       & (0.125)       & (0.015)         & (0.014)       & (0.014)       &  (0.011)      & (0.011)       & (0.011) \\
                        \hline                                                                                                                           
                        $w=0.3$         & -0.113                & 0.126                 & -0.112                 & -0.101                & 0.100                 & -0.100                 &  -0.101       & 0.100         & -0.100 \\                                                                                                                                            
                        & (0.144)       & (0.134)       & (0.140)       & (0.019)         & (0.017)       & (0.017)       &  (0.014)      & (0.013)       & (0.013) \\
                        \hline 
        \end{tabular} }                                                                                                             
        \caption{Post-AGFL estimates, FE and FFE Estimates for DGP's with N=500, T=20 and change points at [T/2] and [2T/3].}     
        \label{tab:case_AGFL}                                                                                                                                      
\end{table} 
\black{In Figure \ref{fig:case4}, we compare the finite sample performance of the AGFL estimator of the number of change points in Qian and Su (2016) with our method for the DGP described above with two change points. To show how the finite sample performance of AGFL and our method depends on the degree of cross-sectional variation, we change the way the regressor is generated to $x_{it} = \sqrt{2}c_i + e_{it}$, with $e_{it} = w g_i + (1-w)\epsilon_{it}$, where $g_i \sim iid~ N(0,1/2)$ and $\epsilon_{it} \sim iid~ N(0,1/2)$. Note that, as before $c_i$ introduces endogeneity in $x_{it}$, while the new variable $g_i$ adds additional exogenous cross-sectional variation to $x_{it}$. The case $w=0$ reflects the usual DGP used throughout this section, which corresponds to $\approx 5 \%$ cross-sectional variation in $e_{it}$, while $w=0.1$ to $\approx 6.5 \%$ and $w=0.3$ to $\approx 20 \%$.\footnote{The cross-sectional variation is the $R$-squared from a regression of $e_{it}$ on the full set of individual dummies.} Figure \ref{fig:case4} shows that even in large samples ($N=500$), the AGFL is very sensitive to this moderate increase in cross-sectional variation, and tends to severely underestimate the number of change points, while our method remains mostly unaffected. Our method is therefore a useful alternative to AGFL for estimating multiple change points in short panels, where most of the variation in the data is in the cross-section dimension.} 

\black{Table \ref{tab:case_AGFL} contrasts the corresponding slope estimates for the two methods, and only for the cases that the two change points at $[T/3]$ and $[2T/3]$ are estimated correctly. The post AGFL slope estimators have higher bias and higher variance, although this could be due to less cases available for simulation averaging. Among the two estimators we propose, we see that as shown in Theorem \ref{theo3}, the FFE estimator is more efficient. }

Since our two applications in the next section have sample sizes $\{T,N\}=\{19,106\}$ and $\{T,N\}=\{18,216\}$ and the second application has 15 regressors, we ran simulations of panels with similar properties. The results in Figures \ref{fig:case5} and \ref{fig:case6} report the estimated number of changes and the estimated location of the change point (when the number of changes is correctly estimated) for a change \black{in each parameter} that varies from 0.01 to 0.02. Our method is able to detect a single change point of moderate size (above $0.015$) at the correct location ($T_1^0=14$) for most of the simulations.

\begin{figure}[htp!]
\begin{minipage}[b]{0.47\linewidth}
\centering
\caption{Estimated number of change points}\label{fig:case5}
\includegraphics[scale=0.6]{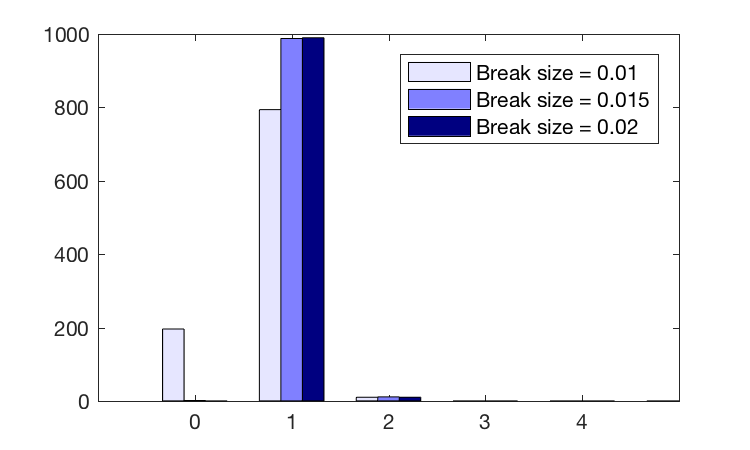}\\
\scriptsize{HQIC with $T=20$, \\ $N=200$, 15 regressors, change point at $T_1^0=14$} \\
\end{minipage}
\begin{minipage}[b]{0.53\linewidth}
\centering
\caption{Estimated location of change points}
\label{fig:case6}
\includegraphics[scale=0.6]{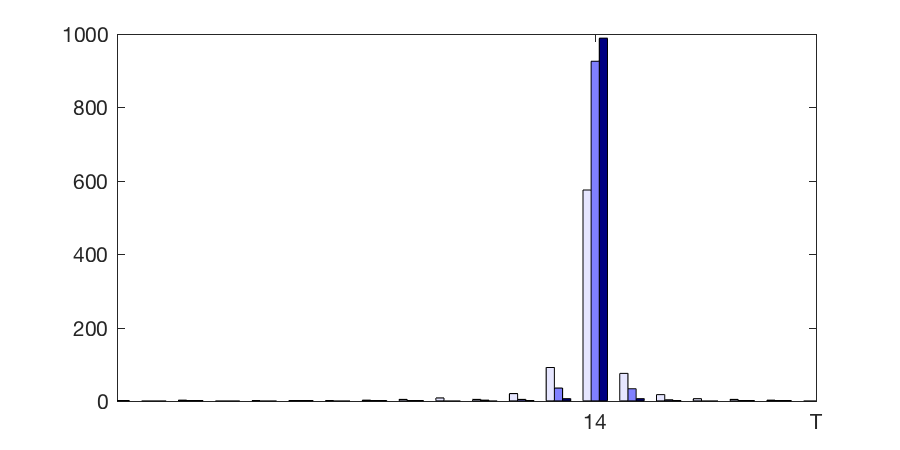}\\
\scriptsize{OLS change point estimator with $T=20$, \\ $N=200$, 15 regressors, change point at $T_1^0=14$} \\
\end{minipage}
\end{figure}

\section{Two Applications}
\label{sec:app}

\subsection{Environmental Kuznets Curve}

The environmental Kuznets curve (EKC) is often used to capture the relationship between income of a country and its emissions of chemicals such as carbon dioxide (CO$_2$). To detect changes in this relationship and in the emissions due to climate accords, we use yearly panel data on 106 countries and 19 years.\footnote{We use part of the dataset from Li, Qian and Su (2016) based on the World Bank Development Indicators (\href{https://data.worldbank.org/products/wdi}{https://data.worldbank.org/products/wdi}).} Countries with population less than five million and countries with missing observations are not used in our analysis. We start in 1992, because that is the year of the UN Framework Convention on Climate Change (UNFCCC), the first large international step to acknowledge climate change  and to attempt to reduce emissions.\footnote{Before this period, several former communist countries would have to be excluded, leading to severe mismeasurement of emissions.}
\begin{equation}
\text{Emissions}_{it} = \beta_{1j} + \beta_{2j} \text{GDP}_{it} +  \beta_{3j} \text{GDP}_{it}^2 +  \beta_{4j} \text{Energy}_{it} + c_{ij} + \varepsilon_{it}
\end{equation}
Here, Emissions$_{it}$ is the logarithm of per capita CO$_2$ emissions in metric tones for country $i$ in year $t$, GDP$_{it}$ represents the logarithm of real gross domestic product in 2000 USD  and Energy$_{it}$ is the logarithm of per capita consumption of energy measured in kilogram of oil equivalent. Energy consumption is included in several applications of the EKC with panel data (see Apergis and Payne 2009, Lean and Smyth 2010, Arouri et al 2012 and Farhani et al 2014). The term $c_{ij}$ reflects unobserved country-specific characteristics affecting CO$_2$ emissions such as its geography, resources, political developments, influence of environmental groups and industry composition, which are all likely to be correlated with income and/or energy consumption. 

\begin{table}[h!]                                                                 
\centering            
\footnotesize                                             
\begin{tabular}{|c|ccccccc|c|}    
\hline
  & \multicolumn{7}{c|}{\textbf{Three Change points: 1997, 2004 and 2007}} & \textbf{No Change}                                                  \\ 
                                & $\hat\beta_1$ & $\hat\beta_2$ & $\hat\beta_3$  & $\hat\beta_4$ &  $\hat\beta_2$-$\hat\beta_1$ & $\hat\beta_3$-$\hat\beta_2$ & $\hat\beta_4$-$\hat\beta_3$ &          $\hat\beta$                            \\              \hline
GDP                     & 0.719***      & 0.039                 & 0.256                 & 0.101           & -0.680*** & 0.217             & -0.156                & 0.585***        \\
                                & (0.111)       & (0.105)       & (0.226)         & (0.209)       & (0.153)       & (0.249)       & (0.308)       & (0.041)         \\
GDP$^2$         & -0.046*** & 0.001             & -0.012                & -0.005          & 0.047***      & -0.013        & 0.007         & -0.035*** \\
                                & (0.008)       & (0.007)       & (0.015)         & (0.013)       & (0.010)       & (0.016)       & (0.020)               & (0.003)         \\
Energy          & 1.271***      & 0.712***      & 0.525***      & 0.981***         & -0.560*** & -0.187*   & 0.456***      & 0.969***      \\
                                & (0.060)       & (0.060)       & (0.126)         & (0.090)       & (0.085)       & (0.140)       & (0.155)               & (0.028)         \\      
                                &                                       &                                         &                                       &                                         &                                       &                                         &                                       &                                         \\ 
Wald Test &                                     &                               &                                 &                               & 73.665*** & 2.493                 & 9.063**       &                               \\      \hline     
N $\times$ T & \multicolumn{8}{c|}{2014}                                                                                                                                                                                                                        \\
N             & \multicolumn{8}{c|}{106}                                                                                                                                                                                                                                        \\      \hline                     
\end{tabular}
\caption{FE estimates of slope coefficients based on yearly data on 106 countries from 1992 to 2010. Three change points are found by  HQIC, at 1997, 2004 and 2007. The last column shows results of a standard fixed effects panel regression where the presence of change points is ignored. Standard errors are reported in brackets and account for autocorrelation and conditional heteroskedasticity (***, **, and * indicate significance at the $0.01$, $0.05$, and $0.1$ level).}                     
\label{tab:results1}                                                        
\end{table}  

Our method finds three change points in 1997, 2004 and 2007.\footnote{Li et al. (2016) study the environmental Kuznets Curve using a similar dataset with an interactive fixed effects specification. Since our estimator only finds evidence for three change points between 1992 and 2010, an interactive fixed effects specification, i.e. $c_{it}=c_i f_t$, may not be desirable, unless it can be argued that despite this specification, there are only three changes in the pseudo-true parameters $\vgam_j^0$.} All these changes can be traced back to steps in the Kyoto Protocol. At the beginning of our sample, the UNFCCC was adopted by 154 countries with the long-term aim of reducing global greenhouse gas emissions. The first major step in this convention was the Kyoto Protocol, adopted by consensus with more than 150 signatories on December 11, 1997. The Protocol included legally binding emissions targets for developed country parties for the six major greenhouse gases (including carbon dioxide). In 2004, Russia and Canada are the last to ratify the Kyoto Protocol, bringing the treaty into effect. In January 2008, the joint implementation mechanism starts. 

In Table \ref{tab:results1} columns 1-4 report the corresponding FE estimates of slope coefficients and columns 5-7 their changes from one segment to the next.\footnote{Table \ref{tab:results1} also reports the Wald test of the $H_0$ hypothesis $\beta_j = \beta_{j-1}$. Moreover, as explained in Section \ref{sec:fe}, the alternative slope estimator (FFE) relies on the assumption $c_{ij}=c_i$. Since several unobserved country characteristics such as the industry composition are likely to vary over segments $j$, and common shocks may hit countries at the same time in an unobserved way, as supported by the interactive fixed effects specification in Li, Qian and Su (2016), the FE estimator is the preferred choice over FFE in this application.} Columns 5-7 of Table \ref{tab:results1} reveal that the three changes are largely driven by changes in the coefficient of energy consumption, which decreases over the second and third segment, then increases back in the last. This could indicate that the Kyoto Protocol was initially successful in decreasing the elasticity of CO$_2$ emissions with respect to energy consumption. Based on estimates in the first sample segment (1992-1997), a 1$\%$ increase in per capita energy consumed leads to a 1.271$\%$ increase in CO$_2$ emissions per capita. The elasticity had decreased to 0.525$\%$ by the third segment (2005-2007). This decrease was followed by a significantly large increase of the elasticity to 0.981$\%$ in the last segment (2008-2010), reversing the decrease over the past 15 years to a large extent.  

In summary, although the Kyoto Protocol is known to have no noticeable impact on global levels of carbon emissions (see, e.g. Helm, 2012), we find that in the course of its implementation, the elasticity of CO$_2$ emissions per capita with respect to energy consumption per capita underwent significant changes.

\subsection{House Price Expectations in the U.S. after the Financial Crisis}

We use data from a quarterly survey on 216 U.S. households to study house price expectations in the aftermath of the subprime mortgage crisis (2009-2013).\footnote{Our dataset is taken from the RAND American Life Panel (ALP), the Office of Federal Housing Enterprise Oversight (\href{http://www.fhfa.gov}{http://www.fhfa.gov}) and the Bureau of Labor Statistics (\href{http://www.bls.gov}{http://www.bls.gov}). We thank G. Niu and A. van Soest who kindly shared the data used in Niu and van Soest (2014), from which we extracted a balanced panel.} Every three months home owners stated their beliefs about the percentage chance that the value of their home will increase by the next year (0-100). We regress these expectations on household characteristics and state-level economic indicators.\footnote{A detailed description of the 15 covariates used can be found in Appendix B.}

The influence of unobserved characteristics of home owner $i$ such as optimism are captured by $c_{ij}$ in model \ref{eq:model}. These unobserved characteristics are likely to be correlated with some of the regressors. For instance, a home owner with an optimistic personality will have a more positive outlook on the price of his house but also on his subjective economic and financial well-being, which is one of the regressors of interest ('economic sentiment'). 

Our method finds a single change point in 2012Q2. In columns 2-7 of Table \ref{tab:results2} we report estimates of slope coefficients for both the FE and the FFE\ estimators. However, as in the first application, the two coefficient estimates exhibit large differences, leading us to conclude that the assumption of fixed effects might be violated in this setting as well. Therefore, we focus on the FE estimates.

Column 7 of Table \ref{tab:results2} shows the difference in estimated FE coefficients between the second and first segment ($\hat \beta_2 - \hat \beta_1$).  Evidently, the change point is primarily driven by differences in coefficients of the variables 'Change in local house prices', being female, the indicator for living in Arizona, California, Florida and Nevada ('Sand state'), and the health of the home owner.\footnote{The Wald test (Table \ref{tab:results2}) of the $H_0$ hypothesis $\beta_1 = \beta_2$ confirms that there is a significant change in the vector of slope coefficients across the two segments for both FFE and FE estimates.} Interestingly, two of these regressors do not vary over time ('Female' and 'Sand state'), yet their coefficient changes.

The period before the change point (2009 - 2012Q2) represents the direct aftermath of the financial crisis when economic uncertainty was high. In this segment we find a significant positive effect of home owner's subjective economic and financial well-being ('economic sentiment') on house price expectations. In 2012Q3 the Federal Reserve announced its third round of quantitative easing, which was an open-ended bond purchasing program of agency mortgage-backed securities and it was announced that the federal funds rate would be likely maintained near zero for at least the next three years. Overall, uncertainty in the market decreased and  the steady recovery of the U.S. housing market began. Our results in Table \ref{tab:results2} suggest that during this recovery period home owners looked for more objective measures of economic performance (such as changes in state-level unemployment rates and state-level house prices) to infer their house values, while previously they relied more on subjective assessments.

\begin{table}[h!]                                                                           
\footnotesize                                                              
\begin{tabular}{|r|c|ccc|ccc|}    
\hline 
 & \textbf{No Change}    & \multicolumn{6}{c|}{\textbf{Single Change at 2012 Q2}} \\ \hline 
 &                                                              & \multicolumn{3}{c|}{} & \multicolumn{3}{c|}{} \\ 
Dependent variable: 1-year & Standard FE & \multicolumn{3}{c|}{Full-sample Fixed Effects (FFE)} & \multicolumn{3}{c|}{sub-sample Fixed Effects (FE)} \\                                                                                              
house price expectations $~~$ & $\hat\beta$ & $\hat\beta_1$ & $\hat\beta_2$ & $\hat{\beta_2} - \hat{\beta_1}$ & $\hat\beta_1$ & $\hat\beta_2$ & $\hat{\beta_2} - \hat{\beta_1}$  \\   \hline   
Constant & - & - & - & 15.158  & - & - & - \\                                                        
  &   &   &   & (13.329) &   &   &   \\                                                              
Change in unemployment & -0.080*** & -0.049 & -0.175** & -0.126* & -0.040 & -0.176** & -0.136* \\    
  & (0.029) & (0.043) & (0.085) & (0.093) & (0.043) & (0.086) & (0.095) \\                           
Change in house prices & 0.631*** & 0.295* & -0.534** & -0.829*** & 0.271* & -0.792*** & -1.063*** \\
  & (0.137) & (0.180) & (0.298) & (0.340) & (0.180) & (0.281) & (0.328) \\                           
Log home value & 0.687** & 0.665 & 1.329** & 0.664  & 0.600 & 1.877** & 1.277* \\                    
  & (0.396) & (0.645) & (0.796) & (0.619) & (0.644) & (0.992) & (0.817) \\                           
Log income per capita & 2.247* & -0.345 & 2.424 & 2.769* & 0.847 & -0.568 & -1.415  \\               
  & (1.650) & (1.800) & (2.102) & (1.698) & (2.455) & (6.804) & (7.276) \\                           
Household size & 0.313 & 0.158 & -1.099 & -1.258  & -0.399 & 0.910 & 1.309  \\                       
  & (1.006) & (1.136) & (1.320) & (1.134) & (1.450) & (4.046) & (4.523) \\                           
Married & 5.351** & 6.662* & 4.786 & -1.877  & 8.838** & 1.086 & -7.752* \\                          
  & (2.894) & (4.643) & (5.559) & (2.443) & (4.656) & (3.085) & (5.199) \\                           
Health & -1.058 & 0.452 & -2.794* & -3.246** & 0.314 & -2.474 & -2.788  \\                           
  & (0.948) & (1.206) & (1.711) & (1.455) & (1.177) & (2.812) & (3.156) \\                           
Non-economic sentiment & 1.644 & 3.077 & 7.843 & 4.766  & 7.893* & -11.589 & -19.482* \\                  
  & (4.200) & (5.806) & (7.925) & (7.288) & (5.841) & (10.901) & (12.696) \\                         
Economic sentiment & 8.188*** & 7.969** & -0.398 & -8.367* & 8.619** & 8.366 & -0.253  \\                 
  & (3.445) & (4.805) & (6.746) & (6.423) & (4.574) & (9.380) & (9.980) \\                           
Age & - & - & - & -0.090  & - & - & - \\                                                             
  &   &   &   & (0.104) &   &   &   \\                                                               
Sand state & - & - & - & 4.732** & - & - & - \\                                                            
  &   &   &   & (2.767) &   &   &   \\                                                               
Female & - & - & - & -7.249*** & - & - & - \\                                                        
  &   &   &   & (2.211) &   &   &   \\                                                               
White & - & - & - & 2.028  & - & - & - \\                                                            
  &   &   &   & (7.262) &   &   &   \\                                                               
Schooling & - & - & - & 3.622* & - & - & - \\                                                        
  &   &   &   & (2.405) &   &   &   \\               \hline                                                
Wald Test &   &   &   & 83.540*** &   &   & 21.697*** \\  
\hline   
N $\times$ T & 3888     & \multicolumn{3}{c|}{3888} & 3024      & 864 & 3888 \\
N                   & 216    & \multicolumn{3}{c|}{216}   & 216   & 216 & 216     \\ 
\hline                               
\end{tabular}                        
\caption{FFE and FE estimates of slope coefficients in the two segments before and after the change point: $\hat \beta_1$ and $\hat \beta_2$. A single change point is found by the OLS change point estimator (with HQIC) at $t=14$ (2012Q2). Column 1 shows results of a standard fixed effects panel regression where the presence of the change point is ignored. Standard errors are reported in brackets and account for autocorrelation and conditional heteroskedasticity (***, **, and * indicate significance levels at $0.01$, $0.05$, and $0.1$).}
\label{tab:results2}                                                        
\end{table}  

\FloatBarrier

\section{Conclusion}
In this paper, we proposed a method for estimating short panels subject to multiple change points and unobserved, possibly time-varying individual effects. We propose first estimating by OLS all the change points that occur in the pseudo-true parameters, under relatively general time-variation in the individual effects. Next, we assume that the individual effects either only change at these identified change points, or remain constant over the sample, and contrast the asymptotic properties of two consistent slope estimators, helping to identify the number and location of changes in the slope parameters.  We demonstrate the usefulness of our method via two applications: the enviromental Kuznets curve and house price expectations.

Our method can also be used as a diagnostic tool for model specification in short panels: if changes are found at each point in time, the specification should be revisited for more parsimonious modelling of time-variation, and if changes occur rarely, then change point modelling is a better alternative.

\section*{Bibliography}
Apergis, N., and Payne, J. E. (2009). {CO2 emissions, energy usage, and output in Central America.} \textit{Energy Policy} 37, 3282-3286.\\
Arouri, M. E. H., Youssef, A. B., M'henni, H., and Rault, C. (2012). {Energy consumption, economic growth and CO2 emissions in Middle East and North African countries.} \textit{Energy Policy} 45, 342-349.\\
Armona, R., Fuster, A. and Zafar, B. (2018). {Home price expectations and behavior: evidence from a
randomized information experiment.} \textit{Review of Economic Studies}, forthcoming.\\
Aue, A. and Horv\'{a}th, L. (2013). {Structural breaks in time series}, \textit{Journal of Time Series Analysis} 34, 1-16.\\
Bai, J. (2009). {Panel data models with interactive fixed effects.} \textit{Econometrica} 77, 1229-1279.\\
Bai, J. (2010). {Common breaks in means and variances for panel data.} \textit{Journal of Econometrics} 157, 78-92.\\
Bai, J. and Li, K. (2014). {Theory and methods of panel data models with interactive effects.} \textit{Annals of Statistics} 42, 142-170.\\
Bai, J., and Perron, P. (1998). {Estimating and testing linear models With multiple structural changes.} \textit{Econometrica} 66, 47-78.\\
Baier, S. L., and Bergstrand, J. H. (2007). {Do free trade agreements actually increase members' international trade?.} \textit{Journal of International Economics} 71(1), 72-95.\\
Baltagi, B.H., Feng, Q. and Kao, C. (2016). {Estimation of heterogeneous panels with
structural breaks}, \textit{Journal of Econometrics} 191, 176-195.\\
Baltagi, B.H., Kao, C. and Liu, L. (2017). {Estimation and identification of change points in panel models with nonstationary or stationary regressors and error term}, \textit{Econometric Reviews 36: 85-102.}\\
Bardwell, L., Fearnhead, P., Eckley, I.A., Smith, S. and Spott, M. (2018). {Most recent changepoint detection in panel data.} \textit{Technometrics}, \href{https://doi.org/10.1080/00401706.2018.1438926}{https://doi.org/10.1080/00401706.2018.1438926}/\\
Blanco, L. and Ruiz, I. (2013). {The Impact of Crime and Insecurity on Trust in Democracy and Institutions,} \textit{American Economic Review: Papers and Proceedings} 103,  284-288.\\
Börsch-Supan, A., Brandt, M., Hunkler, C., Kneip, T., Korbmacher, J., Malter, F., Schaan, B., Stuck, S. and Zuber, S. (2013). {Data resource profile: the Survey of Health, Ageing and Retirement in Europe (SHARE).} \textit{International Journal of Epidemiology} 42, 992-1001.\\
Chan, J., Horv\'{a}th, L., and Hu\v{s}kov\'{a}, M. (2013). {Darling-Erd\H{o}s limit results for change point detection in panel data.} \textit{Journal of Statistical Planning and Inference} 143, 955-970.\\
Chan, N. H., Yau, C. Y., and Zhang, R. (2014). {Group LASSO for Structural Break Time Series.} \textit{Journal of American Statistical Association} 109, 590-599.\\
Cs\"{o}rg\"{o}, M., and Horv\'{a}th, L. (1997). {Limit theorems in change point Analysis.} \textit{Wiley Series in Probability and Statistics} Vol. 18. John Wiley \& Sons Inc.\\
Cho. H. (2016). {change point detection in panel data via double CUSUM statistic.} \textit{Electronic Journal of Statistics} 10, 2000-2038.\\
Cho, H. and Fryzlewicz, P. (2015). {Multiple change point detection for high dimensional time series via sparsified binary segmentation.} \textit{ Journal of the Royal Statistical Society B} 77, 475-507.\\
de Wachter, S. and Tsavalis, E. (2012). {Detection of structural breaks in linear dynamic panel data models}, \textit{Computational Statistics and Data Analysis} 56, 3020-3034. \\
Emerson, J. and Kao, C. (2001). {Testing for structural change of a time trend regression in panel data: Part I.} \textit{Journal of Propagations in Probability and Statistics} 2, 57-75.\\
Emerson, J. and Kao, C. (2002). {Testing for structural change of a time trend regression in panel data: Part II.} \textit{Journal of Propagations in Probability and Statistics} 2, 207-250.\\
Farhani, S., Mrizak, S., Chaibi, A., and Rault, C. (2014). {The environmental Kuznets curve and sustainability: A panel data analysis.} \textit{Energy Policy} 71, 189-198.\\
Geng, N. and van Soest, A. H. O. (2014). {House price expectations.} \textit{IZA Working Paper 8536}, \href{http://ftp.iza.org/dp8536.pdf}{http://ftp.iza.org/dp8536.pdf}.\\
Hall, A. R., Osborn, D. and Sakkas, N. (2013). {Inference about structural breaks using information criteria}, \textit{The Manchester School} 81, 54-81.\\
Harchaoui, Z., and L\'{e}vy-Leduc, C. (2010). {Multiple change point estimation with a total variation enalty.} \textit{Journal of the American Statistical Association} 105, 1480-1493.\\
Helm, D. (2012). {Climate policy: the Kyoto approach has failed.} \textit{Nature} 491, 663-665. \\
Horv\'ath, L., and Hu\v{s}kov\'a, M. (2012). {change point detection in panel data.} \textit{Journal of Time Series Analysis} 33, 631-648.\\
Kim, D. (2011). {Estimating a common deterministic time trend break in large panels with cross sectional dependence.} \textit{Journal of Econometrics} 164(2), 310-330.\\
Kyle, M. and Williams, H. (2016). {Is American health care uniquely inefficient? Evidence from prescription drugs.} \textit{American Economic Review} 107, 486-490.\\
Lean, H. H., and Smyth, R. (2010). {CO2 emissions, electricity consumption and output in ASEAN.} \textit{Applied Energy} 87, 1858-1864.\\
Li, D., Qian, J. and Su, L. (2016). {Panel data models with interactive fixed effects and multiple structural breaks.} \textit{Journal of the American Statistical Association}, 516: 1804-1819.\\
Moon, H. R., and Weidner, M. (2015). {Linear regression for panel With unknown number of factors as interactive fixed effects.} Econometrica 83, 1543-1579.\\
Nimomiya, Y. (2005). {Information criterion for Gaussian change point model.} \textit{Statistics and Probability Letters} 72, 237-247.\\
Perron, P. and Yamamoto, Y. (2015). {Using OLS to estimate and test for structural changes in models with endogenous regressors.} \textit{Journal of Applied Econometrics} 28, 119-144. \\
Pesaran, H. (2006). {Estimation and inference in large heterogeneous panels with a multifactor error structure.} \textit{Econometrica} 74, 967-1012.\\
Okuy, R. and Wang, W. (2018). {"Heterogeneous structural breaks in panel data models.} \textit{Working Paper}, \href{"Heterogeneous structural breaks in panel data models}, \href{https://papers.ssrn.com/sol3/papers.cfm?abstract_id=3031689}{https://papers.ssrn.com/sol3/papers.cfm?abstract$\_$id=3031689}.\\
Qian, J. and Su, L. (2016). {Shrinkage estimation of common breaks in panel data models via adaptive group fused LASSO.} \textit{Journal of Econometrics} 191, 86-109.\\
Qu, Z., and Perron, P. (2007). {Estimating and testing structural changes in multivariate regressions.} \textit{Econometrica} 75, 459-502.\\
Torgovitski, L. (2015). {Panel data segmentation under finite time horizon.} \textit{Journal of Statistical Planning and Inference} 167, 69-89.\\
Vert, J.-P. and Bleakley, K. (2010). {Fast detection of multiple change points shared by many signals using group LARS.} \textit{Proceedings of Advances in Neural Information Processing Systems 23.}\\

\section*{Appendix A}
\begin{proof}[Proof of Theorem \ref{theo1}]

 \textbf{
Part (i).} For simplicity and for this part of the proof only, drop the $m^0$ subscripts on all quantities and the hat on $\hat \vlam$. Recall that $\hat \vbeta_{\vlam} = (\tilX' \tilX)^{-1} \tilX'\vy$ is the OLS estimator using the partition $\vlam$. Let $\vvc = \vc_{t=1:T} (\vc_{i=1:N} (c_{it}))$, $\veps= \vc_{t=1:T} (\vc_{i=1:N} (\epsilon_{it}))$, $\vgam_{\vlam^0} = \vbeta^0 + (\tilX^0 \tilX^0)^{-1} \tilX^{0'} \vvc$, $\vP_{\tilX^0} =\tilX^0 (\tilX^{0'} \tilX^0)^{-1}\tilX^{0'}$ and $\vM_{\tilX^0} = \vI_{NT}-\vP_{\tilX^0}$, with $\vI_{NT}$  the $NT \times NT$ identity matrix. With this notation,
\begin{align*}
S_{NT}( \vgam_{\vlam^0},\vlam^0) &= (NT)^{-1} ||\vy - \tilX^0 \vgam_{\vlam^0}||^2 =(NT)^{-1}||\veps + \vvc + \tilX^0(\vbeta^0- \vgam_{\vlam^0})||^2 \\
&= (NT)^{-1} \veps'\veps + 2(NT)^{-1}  \veps'\vvc + 2 (NT)^{-1}\veps' \tilX^0 \vbeta^0 \\
&-2(NT)^{-1}  \veps' \tilX^{0} (\tilX^{0'} \tilX^0)^{-1} \tilX^{0'}\vvc + ||\vvc+\tilX^0(\vbeta^0- \vgam_{\vlam^0})||^2.
\end{align*}
By A\ref{a1}(i)-(iv), $ 2(NT)^{-1}  \veps'\vvc =\op(1)$,   $2 (NT)^{-1}\veps' \tilX^0 =\op(1)$, $(NT)^{-1}\tilX^{0'} \tilX^0 =\Op(1)$, $(NT)^{-1} \tilX^{0'}\vvc =\Op(1)$. Also, $\tilX^0(\vbeta^0- \vgam_{\vlam^0}) = -\vP_{\tilX^0} \vvc$ and $\vvc+\tilX^0(\vbeta^0- \vgam_{\vlam^0}) =\vM_{\tilX^0} \vvc$ . Therefore, 
\begin{align}\label{eq1} S_{NT}( \vgam_{\vlam^0},\vlam^0) = (NT)^{-1} \veps'\veps + \vvc'\vM_{\tilX^0} \vvc +\op(1).  
\end{align}
On the other hand, we have:
        \begin{align} \nonumber
        &S_{NT}( \hat \vbeta_{\vlam},\vlam) - S_{NT}( \vgam_{\vlam^0},\vlam^0) =  (NT)^{-1} (\vy-\tilX \hat \vbeta_{\vlam})'(\vy-\tilX \hat \vbeta_{\vlam}) \\\nonumber
&=         (NT)^{-1} (\veps + \vvc + \tilX^0\vbeta^0-\tilX \hat \vbeta_{\vlam})' (\veps+ \vvc + \tilX^0\vbeta^0-\tilX \hat \vbeta_{\vlam})\\\nonumber
&=(NT)^{-1} \veps'\veps + 2(NT)^{-1} \veps' (\vvc+ \tilX^0\vbeta^0-\tilX \hat \vbeta_{\vlam}) + (\vvc + \tilX^0\vbeta^0-\tilX \hat \vbeta_{\vlam})'(\vvc + \tilX^0\vbeta^0-\tilX \hat \vbeta_{\vlam})\\\nonumber
&-[(NT)^{-1} \veps'\veps + \vvc'\vM_{\tilX^0} \vvc] +\op(1)\\\nonumber
& = 2(NT)^{-1} \veps' (\vvc+ \tilX^0\vbeta^0-\tilX \hat \vbeta_{\vlam}) + (\vvc + \tilX^0\vbeta^0-\tilX \hat \vbeta_{\vlam})'(\vvc + \tilX^0\vbeta^0-\tilX \hat \vbeta_{\vlam})-(NT)^{-1} \vvc'\vM_{\tilX^0} \vvc \\
&+\op(1)  \equiv 2 I + II -III +\op(1).\label{eq2}
            \end{align}
By definition, $S_{NT}( \hat \vbeta_{\vlam},\vlam) - S_{NT}(\vgam_{\vlam^0},\vlam^0) \leq 0$. Therefore,
$
S_{NT}( \hat \vbeta_{\vlam},\vlam) -  S_{NT}(\vgam_{\vlam^0},\vlam^0) = 2I\ + II-III+\op(1) \leq 0.$
We prove consistency of $\vlam$ (and therefore of the break-point estimators $\hat T_j$) in two steps. In step 1, we show that $I=\op(1)$ and $II-III= \Op(1)$, meaning that $II-III$ asymptotically dominates $I$, so $\plim(II- III) \leq 0$. In step 2, we prove consistency by contradiction; if there is at least one change point estimator that is not consistent for its true value, then $\plim( II-III)>C$ for some $C>0$, contradicting $\plim(II-III) \leq 0$.\\
\textbf{
Step 1.}
For any partition $\vlam$, 
\begin{align*}
&\tilX \hat \vbeta_{\vlam}= \tilX (\tilX' \tilX)^{-1} \tilX' (\tilX^0 \vbeta^0 + \vvc) + \op(1) \\
& \Leftrightarrow  \vvc + \tilX^0\vbeta^0-\tilX \hat \vbeta_{\vlam} = (\vI_{NT} - \tilX(\tilX' \tilX)^{-1} \tilX')(\vvc +\tilX^0\vbeta^0) +\op(1) \\
&\Leftrightarrow II = (NT)^{-1} (\vvc +\tilX^0\vbeta^0)' (\vI_{NT} - \tilX(\tilX' \tilX)^{-1} \tilX') (\vvc +\tilX^0\vbeta^0) \\
& = (NT)^{-1} (\vvc +\tilX^0\vbeta^0)'(\vvc +\tilX^0\vbeta^0) \\ \nonumber
&- (NT)^{-1} (\tilX' \vvc +\tilX' \tilX^0\vbeta^0)'[(NT)^{-1} \tilX' \tilX]^{-1}(NT)^{-1} (\tilX' \vvc +\tilX' \tilX^0\vbeta^0)\\
& \equiv II_A - II_B.
\end{align*}
Letting $\Delta \lambda_j^0 = \lambda_j^0-\lambda_{j-1}^0$, by A\ref{a1}(iii)-(vi), 
\begin{align*}
II_A &= (NT)^{-1}\vvc'\vvc + 2 (NT)^{-1}\vbeta^{0'} \tilX^{0'} \vvc + (NT)^{-1}\vbeta^{0'} (\tilX^{0'}\tilX^0)^{-1} \vbeta^0 \\
&=\sigma^2_c + \vbeta^{0'} \vc_{1:m^0+1} (\Delta \lambda_j^0 \va_j^0) + \sum_{j=1}^{m^0+1} (\Delta \lambda_j^0)^{-1} \vbeta_j^{0'} (\vQ_j^0)^{-1} \vbeta_j^0 +\op(1)= \Op(1).
\end{align*}
By A\ref{a1}(iii),
\begin{align*}
(NT)^{-1}\tilX '\vvc &= \vc_{j=1:m^0+1}\left((NT)^{-1} \sum_{s=1}^{m^0+1}\sum_{t \in (I_j\cap I_s^0)} \vx_{it}c_{it}\right)  +\op(1)\\
& = \vc_{1:m^0+1}(\sum_{s=1}^{m^0+1} \lambda_{js} \va_s^0)+\op(1), \end{align*}
where $\lambda_{js}$ is the cardinality of the set ($I_j\cap I_s^0$) - recall that $I_j$ is the estimated regime -  divided by the sample size $T$. Similarly, by A\ref{a1}(iv), $
(NT)^{-1}\tilX '\tilX^0 \vbeta^0\inp \vc_{1:m^0+1}(\sum_{s=1}^{m^0+1} \lambda_{js} \vQ_s^0 \vbeta_s^0).$ Therefore, 
$
(NT)^{-1}\tilX '( \vvc + \tilX^0 \vbeta^0) \inp \vc_{j=1:m^0+1}\left[\sum_{s=1}^{m^0+1} \lambda_{js} (\va_s^0+ \vQ_s^0 \vbeta_s^0)\right] = \Op(1).
$
Also, by A\ref{a1}(iv), it can be shown that $(NT)^{-1} \tilX' \tilX = \diag_{j=1:m^0+1} \sum_{s=1}^{m^0+1} \lambda_{js} \vQ_s^0 = \Op(1)$. Therefore, $II_B = \Op(1)$, and 
\begin{align}\label{eq3}
II = II_A -II_B = \Op(1).
\end{align}

By A\ref{a1}(iii), (iv) and (vi), $III = (NT)^{-1} \vvc' \vM_{\tilX^0} \vvc = \sigma_{c,T}^2+ \sum_{j=1}^{m^0+1} \Delta \lambda_j^0 \va_{j}^{0'} \vQ_j^0 \va_j^0 = \Op(1)$. Therefore, 
$
II-III = \Op(1).
$ 
On the other hand, using some of the above results and A\ref{a1}(ii),
\begin{align}\nonumber
I & = (NT)^{-1} \veps' \vvc+ [(NT)^{-1}\veps'\tilX^0 ]\vbeta^0 -(NT)^{-1}\veps'\tilX \hat \vbeta_{\vlam} = \op(1) + \op(1) \\
&- \op(1) [ (\tilX' \tilX)^{-1} \tilX' \tilX^0 \vbeta^0 +(\tilX' \tilX)^{-1} \tilX' \vvc +\op(1)] = \op(1)[ \Op(1)+ \Op(1)+\op(1)] = \op(1).\label{eq4}
\end{align}
Hence, $II-III$ dominates $I$ in probability order, so $\plim(II-III) \leq 0$.

\textbf{Step 2.} We now show that the change point estimators are consistent by contradiction. Suppose that $m^0 < T-1$, else there are change points at each sample period, so by default, all the estimated change points are equal to the true change points. If $m^0< T-1$, and not all the change points are equal to the true change points with probability one in the limit, then there is at least one estimated regime $j$ that contains a true change point. Formally, there exists $k\in \{1,\ldots, m^0\}$ and $j\in \{1,\ldots, m^0+1\}$, such that $\hat T_{j-1} < T_k^0 < \hat T_{j}$ (here $\hat T_0=0$ and $\hat T_{m^0+1}=T$). Therefore, in both periods  $T_k^0$ (belonging to the true regime $I_k^0$) and $T_{k}^0+1$ (belonging to the true regime $I_{k+1}^0$), we are estimating $\hat \vbeta_j$ (note that here we can take more than one period but for simplicity we only consider one). Let $\vX_{t}$ be the $N\times p$ matrix with rows $\vx_{t}'$, let $\vvc_{t}$ be the $N\times 1$ vector with elements $c_{it}$, let $\hat \vbeta_t = \hat\vbeta_j$ in interval $[\hat T_{j-1}, \hat T_j]$, let $\vP_{\vX_{t}} =\vX_{t} (\vX_{t}'\vX_{t})^{-1} \vX_{t}'$ and $\vM_{\vX_t} = \vI_{N}-\vP_{\vX_{t}}$, and denote $\overline \vX_{k+s} =\vX_{T_k^0+s}$ for $s=0,1$. Also, let $\vbeta_t^0 = \sum_{k=1}^{m^0+1} \vbeta_j^0  1\{t \in I_k^0\}$ and $\hat \vbeta _t = \sum_{j=1}^{m^0+1} \hat \vbeta_j  1\{t \in [\hat T_{j-1}+1,\hat T_j]\}$. Then:
\begin{align*}
II & =   (NT)^{-1} \sum_{t=1}^T ||\vvc_t + \vX_{t} \vbeta_{t}^0 -\vX_t \hat \vbeta_t)||^2 \\
& =  (NT)^{-1} \sum_{t=1, t\neq T_{k}^0, t\neq T_{k}^0+1}^T ||\vvc_t + \vX_{t} \vbeta_{t}^0 -\vX_t \hat \vbeta_t)||^2 \\
&+
(NT)^{-1} \sum_{s=0}^1  ||\vvc_{k+s} + \overline \vX_{k+s} \vbeta_{k+s}^0 - \overline \vX_{k+s} \hat \vbeta_j||^2 \\
& \geq (NT)^{-1} \sum_{t=1}^T(\vvc_t + \vX_{t} \vbeta_{t}^0 -\vX_t \hat \vbeta_t)' \vM_{\vX_t}(\vvc_t + \vX_{t} \vbeta_{t}^0 -\vX_t \hat \vbeta_t) \\
&+ \sum_{s=0}^1 (NT)^{-1} (\vvc_{k+s} + \overline \vX_{k+s} \vbeta_{k+s}^0 - \overline \vX_{k+s} \hat \vbeta_j)' \vP_{\overline \vX_{k+s}} (\vvc_{k+s} + \overline \vX_{k+s} \vbeta_{k+s}^0 - \overline \vX_{k+s} \hat \vbeta_j)\\
&=(NT)^{-1} \sum_{t=1}^T\vvc_t' \vM_{\vX_t} \vvc_t \\
&+ \sum_{s=0}^1 (NT)^{-1} (\vvc_{k+s} + \overline \vX_{k+s} \vbeta_{k+s}^0 - \overline \vX_{k+s} \hat \vbeta_j)' \vP_{\vX_{k+s}} (\vvc_{k+s} + \overline \vX_{k+s} \vbeta_{k+s}^0 - \overline \vX_{k+s} \hat \vbeta_j)\\
&  = (NT)^{-1} \sum_{t=1}^T\vvc_t' \vM_{\vX_t} \vvc_t\\
&+T^{-1} \sum_{s=0}^1 ( N^{-1} \overline \vX_{k+s}'\overline \vX_{k+s})^{-1} (N^{-1} \overline \vX_{k+s}'\vvc_{k+s} + N^{-1} \overline \vX_{k+s}'\overline \vX_{k+s} \vbeta_{k+s}^0 - N^{-1} \overline \vX_{k+s}'\overline \vX_{k+s} \hat \vbeta_j) ' \\
& \times ( N^{-1} \overline \vX_{k+s}'\overline \vX_{k+s}) \\
& \times (N^{-1}\vX_{k+s}'\overline \vX_{k+s})^{-1} (N^{-1}\overline \vX_{k+s}'\vvc_{k+s} + N^{-1} \overline \vX_{k+s}'\overline \vX_{k+s} \vbeta_{k+s}^0 - N^{-1} \overline \vX_{k+s}'\overline \vX_{k+s} \vbeta_{k+s}^0 ),
\end{align*}
where the first inequality holds because $\bm d' \bm d \geq \bm d' \vM \bm d  $ for any vector $d$ and projection matrix $\vM$. Note that $N^{-1} \vX_{t}'\vX_{t} \inp \vQ_{j}^0$ for $t \in I_j^0$ by A\ref{a1}(iv), and $N^{-1} \vX_{t}'\vvc_{t} \inp \va_j^0$ for $t\in I_j^0$ by A\ref{a1}(iii). It follows  that $T^{-1} \sum_{t=1}^T\vvc_t' \vM_{\vX_t} \vvc_t= (NT)^{-1} \vvc' \vvc - \sum_{j=1}^{m^0+1} T^{-1} \sum_{t=1}^T \va_j^{0'} (\vQ_{j}^0)^{-1} \va_j^{0'}1[t\in I_j^0] +\op(1) = (NT)^{-1} \vvc' \vvc - \sum_{j=1}^{m^0+1} \Delta \lambda_j^0 \va_j^{0'} (\vQ_{j}^0)^{-1} \va_j^0 +\op(1)$. 

Note that $III= (NT)^{-1} \vvc'\vM_{\tilX^0} \vvc = (NT)^{-1} \vvc'\vvc - \sum_{j=1}^{m^0+1} (\lambda_j^0 \va_j^0)' (\lambda_j^{0} \vQ_j^0)^{-1} (\lambda_j^0 \va_j^0) +\op(1)$. Therefore, $(NT)^{-1} \sum_{t=1}^T\vvc_t' \vM_{\vX_t} \vvc_t - III =\op(1)$.

By A\ref{a1}(iv), $N^{-1} \overline \vX_{k+s}'\overline \vX_{k+s} \inp \vQ_{k+s}^0$ and by A\ref{a1}(iii),  $ N^{-1} \overline \vX_{k+s}'\vvc_{k+s} + N^{-1} \overline \vX_{k+s}'\overline \vX_{k+s} \vbeta_{k+s}^0 \inp  \va_{k+s}^0 + \vQ_{k+s}^0 \vbeta_{k+s}^0$. This also implies
$ ( N^{-1} \overline \vX_{k+s}'\overline \vX_{k+s})^{-1} (N^{-1} \overline \vX_{k+s}'\vvc_{k+s} + N^{-1} \overline \vX_{k+s}'\overline \vX_{k+s} \vbeta_{k+s}^0) \inp (\vQ_{k+s}^0)^{-1} \va_{k+s}^0 + \vbeta_{k+s}^0 = \vgam_{k+s}^0$. Letting $\eta_{k+s} = \min \mbox{eig } \vQ_{k+s}^0 >0$ by A\ref{a1}(iv), and noting that $|| Q_{k+s}^0|| > \eta_{k+s}$ by the definition of the matrix norm in the notation section, we have:
\begin{align*}
II-III &\geq \op(1)+ T^{-1} \sum_{s=0}^1 (\vgam_{k+s}^0 -\hat \vbeta_j)' \vQ_{k+s}^0 (\vgam_{k+s}^0 - \hat \vbeta_j) \geq \op(1)+ T^{-1} \sum_{s=0}^1 \eta_{k+s} ||\vgam_{k+s}^0 -\hat \vbeta_j||^2 \\
& \geq \op(1)+ T^{-1} 0.5 \min (\eta_{k},\eta_{k+1}) ||\vgam_{k+1}^0 - \vgam_{k}^0||^2 >C +\op(1),
\end{align*}
for some $C>0$, where the second to last inequality follows from $||a-c||^2+||b-c||^2 \geq 0.5 ||b-c||^2$, for all same size vectors $a,b,c$, and the last inequality from A\ref{a1}(vi). Since $II-III>C+\op(1)$ or $\plim(II-III)>C$ contradicts $\plim(II-III)=\op(1)$, all estimators are consistent, i.e. $\lim_{N \rightarrow \infty} P(\hat T_i = T_i^0) = 1$. Therefore, $\hat \vbeta_{\vlam} = \hat \vbeta_{\vlam^0} +\op(1) = \vgam_{\vlam^0}^0 + (\tilX^{0'}\tilX^0)^{-1} \tilX^{0'} \veps =  \vgam_{\vlam^0}^0+ \op(1)$, implying that $\hat \vbeta_{j,\vlam} \inp \vgam_j^0$.\\
For parts (ii)-( iii) of the proof, we reintroduce the $m,m^0$ subscripts to indicate different numbers of change points, and the hats on the estimated partitions ($\hat \vlam$).\\
\textbf{Part (ii). } 

\textbf{Case 1. $m=m^0$}. From part (i), all change point estimators are equal to the true change points in the limit, and each element of the sum of squared residuals is $\Op(1)$. Therefore, for deriving limits of the sum of squared residuals, we can proceed as if the estimators were equal to the true change points.  By the proof of Theorem 1 and A\ref{a1}(vi), it can be shown that:
\begin{align*}
S_{NT}(\hat \vbeta_{\hat \vlam_{m^0}},\vlam_{m^0}) &= S_{NT}(\vgam_{\vlam_{m^0}}^0 ,\vlam_{m^0}^0) + \op(1) = (NT)^{-1} \veps'\veps+(NT)^{-1} \vvc'\vM_{\tilX_{m^0}^0} \vvc + \op(1) \\
& = \sigma^2_{\epsilon, T} + \sigma^2_{c,T}-\sum_{j=1}^{m^0+1} \Delta \lambda_j^0 \va_j^{0'} (\vQ_{j}^0)^{-1} \va_j^0 +\op(1).
\end{align*}
\textbf{Case 2. $m>m^0$}. First, suppose out of the $m$ change point estimators,  there are exactly $m^0$ that are consistent, i.e. those are $\hat T_j$, $j \in\{1,\ldots, m\}$ such that $\lim_{N\rightarrow \infty} P(\hat T_j =T_s^0) =1$. Then, assuming the sample size is large, we can proceed as if $\hat T_j=T_s^0$, since the sum of squared residuals for $m$ change points is $\Op(1)$ regardless of the partition considered (examine equations \eqref{eq2}-\eqref{eq4} and note that $I=\op(1)$, $II=\op(1)$ and $III=\Op(1)$ regardless of the number of change points imposed in estimation.) Additionally, there are $0\leq m_j<m-m^0$ change point estimators in each true regime $I_j^0$, so $m_j+1$ sub-regimes of $I_j^0$, which we denote by $I_{js}$ for $s=1,\ldots, m_j+1$. Denote by $\lambda_{js}, (s=1,\ldots, m_j)$ the associated change point estimators divided by $T$, and by $\hat \vbeta_{js}, (s=1,\ldots, m_{j}+1)$ the corresponding parameter estimators in each sub-regime of $I_j^0$, if $m_j \neq 0$. Also let $\vbeta_{\vlam_{m^0}^0} \equiv \vc_{1:m^0+1}(\hat \vbeta_j^0)$, where $\hat \vbeta_j^0$ refers to the estimators $\vbeta_j$ that use the true regime $I_j^0$. Then, letting $\Delta \lambda_{js} = \lambda_{js}-\lambda_{js-1}$ (where $\lambda_{j0}=\lambda_j^0$, $\lambda_{j m_j+1}=\lambda_{j+1}^0$), we have:
\begin{align*}
&(NT) [ S_{NT}(\hat \vbeta_{\hat \vlam_{m}},\vlam_{m}) -  S_{NT}(\hat \vbeta_{\vlam_{m^0}^0},\vlam_{m^0}^0)] \\
&=  \sum_{j=1}^{m^0+1} \mathbf 1[ m_j \neq 0] \sum_{s=1}^{m_j+1} \sum_{t\in I_{js}} \sum_{i=1}^N  (y_{it}- \vx_{it}'\hat \vbeta_{js})^2 - (y_{it}- \vx_{it}'\hat \vbeta_{j}^0)^2 \\
& =  \sum_{j=1}^{m^0+1} \mathbf 1[ m_j \neq 0] \sum_{s=1}^{m_j+1} (\hat \vbeta_{js}-\hat \vbeta_{j}^0)' \sum_{t\in I_{js}} \sum_{i=1}^N [2 \vx_{it}u_{it} - \vx_{it} \vx_{it}'(\hat \vbeta_{js}-\vbeta_j^0) - \vx_{it} \vx_{it}'(\hat \vbeta_{j}^0-\vbeta_j^0)] \\
& =   \sum_{j=1}^{m^0+1} \mathbf 1[ m_j \neq 0] \sum_{s=1}^{m_j+1} \sqrt{NT} (\hat \vbeta_{js}-\hat \vbeta_{j}^0)' \Delta \lambda_{js} \vQ_j^0 \sqrt{NT} (\hat \vbeta_{js}-\hat \vbeta_j^0) +\op(1),
\end{align*}
where the last equality uses the results $(NT)^{-1}\sum_{t\in I_{js}} \vx_{it} u_{it} = (\Delta \lambda_{js} \vQ_j^0)(\hat \vbeta_{js}-\hat \vbeta_j^0) +\op(1)$ and $(NT)^{-1}\sum_{t\in I_{js}} \vx_{it} \vx_{it}' =\Delta \lambda_{js} \vQ_j^0 +\op(1)$. Recalling that $\vgam_j^0 = \vbeta_j^0+ (\vQ_j^0)^{-1} \va_j^0$, it can be shown by standard arguments that $\Delta \lambda_j^0 (\hat \vbeta_j^0 - \vgam_j^0) = \sum_{s=1}^{m^0+1} \Delta \lambda_{js} (\hat \vbeta_{js}^0 -\vgam_j^0)+\op(1)$ and that $\sqrt{NT} \Delta \lambda_j^0 (\hat \vbeta_j^0 - \vgam_j^0) =\Op(1)$ (using A1(iv)). Therefore, $\sqrt{NT} (\hat \vbeta_{js}-\vgam_j^0)=\Op(1)$, and so $\sqrt{NT} (\hat \vbeta_{js}-\hat \vbeta_j^0) =\Op(1)$, and therefore
$
   (NT) [ S_{NT}(\hat \vbeta_{\hat \vlam_{m}},\vlam_{m}) - (NT)^{-1} S_{NT}(\hat \vbeta_{\vlam_{m^0}^0},\vlam_{m^0}^0)] =\Op(1).
$
Using $\log(1+a) = a-a^2/2+a^3/3+...$,
\begin{align*}
&\log [S_{NT}(\hat \vbeta_{\hat \vlam_{m}},\vlam_{m})] - \log [S_{NT}(\hat \vbeta_{\vlam_{m^0}^0},\vlam_{m^0}^0)] = \log \left[1+ \frac{(NT)[ S_{NT}(\hat \vbeta_{\hat \vlam_{m}},\vlam_{m}) -  S_{NT}(\hat \vbeta_{\vlam_{m^0}^0},\vlam_{m^0}^0)]}{ (NT) S_{NT}(\hat \vbeta_{\vlam_{m^0}^0},\vlam_{m^0}^0)}\right] \\
&= \log \{1+ \Op[(NT)^{-1}]\} = \Op[(NT)^{-1}].
\end{align*}

Therefore, in this case, for $m>m^0$,
$IC(m)- IC(m^0) = \Op[(NT)^{-1}] + (p^*_m-p^*_{m^0}) \ell_{NT}$, and $p^*_m-p^*_{m^0}>0$ and dominates the $\Op(NT)^{-1}$ term. Therefore, $\plim[IC(m)>IC(m^0)]=1$.

Now suppose that we have at least $m-m^0+1$ change point estimators that are different than the true ones in the limit, so that not all change points can be paired with consistent estimators. Noting that Step 2 of the Proof of part (i) does not depend on the number of breaks, the same arguments as in Step 2 of part (i) can be employed to show that for some constant $C>0$, $S_{NT}(\hat \vbeta_{\hat \vlam_{m}},\vlam_{m})] > C+  S_{NT}(\hat \vbeta_{\vlam_{m^0}^0},\vlam_{m^0}^0) +\op(1)$. Therefore,
\begin{align*}
\plim \{\log [S_{NT}(\hat \vbeta_{\hat \vlam_{m}},\vlam_{m})] - \log [S_{NT}(\hat \vbeta_{\vlam_{m^0}^0},\vlam_{m^0}^0)]\} > log (1+C).
\end{align*}
So, $IC(m)-IC(m^0)> \log(1+C) +(p^*_m-p^*_{m^0}) \ell_{NT} +\op(1)> \log(1+C) +\op(1)$, so $\plim[IC(m)>IC(m^0)]=1$, which concludes the proof for $m>m^0$.

\textbf{Case 3. $m<m^0$}. Here, necessarily $m^0-m$ changes are not estimated at all, so there must be at least one true change point that is skipped in the estimation. Then by the same arguments as in Step 2 of part(i), and in the last two paragraphs above, there is a $C>0$ such that $IC(m)-IC(m^0)> \log(1+C) + (p^*_m-p^*_{m^0}) \ell_{NT} +\op(1)> \log(1+C)+\op(1)$, so $\plim[IC(m)>IC(m^0)]=1$, which concludes the proof for $m<m^0$. 

\textbf{
Part (iii). }Note that because $ \plim [S_{NT}(\hat \vbeta_{\hat \vlam_{m}},\vlam_{m}) - (NT)^{-1} S_{NT}(\hat \vbeta_{\vlam_{m^0}^0},\vlam_{m^0}^0)] =\op(1)$ for all the change points being consistently estimated (and some more), while we have that $\plim [S_{NT}(\hat \vbeta_{\hat \vlam_{m}},\vlam_{m}) - (NT)^{-1} S_{NT}(\hat \vbeta_{\vlam_{m^0}^0},\vlam_{m^0}^0)] >C$ otherwise, the sum of squared residuals is smallest in the limit when all change points are consistently estimated. Therefore, if we impose more than $m^0$ change points, all the change points can be paired with consistent estimators.
\end{proof}

\begin{proof}[Proof of Lemma \ref{lem1}] \hfill \\
\textbf{A\ref{a1}(i).} Let $\bm \omega_i = \sum_{t=1}^T \vx_{it} \epsilon_{it}$. By A\ref{a2}(i),  $E(\bm \omega_i)=0$. By A\ref{a2}(ii), using the triangle inequality and H\"older's inequality, $||\bm \omega_i||_{2+\delta/2} \leq   \sum_{t=1}^T \| \vx_{it} \|_{4+\delta} \, \|\epsilon_{it}\|_{4+\delta} \leq T \sup_t \|  \vx_{it} \|_{4+\delta} \, \sup_t \|\epsilon_{it}\|_{4+\delta} < \infty$. By A\ref{a2}(iii), $E(\bm \omega_i \bm \omega_i')= E\left( \sum_{t,s=1}^T \vx_{it}\vx_{is}\epsilon_{it}\epsilon_{is}\right) = \sum_{t,s=1}^T \bm V_{ts}^0$. Therefore, by the CLT, $N^{-1/2} \sum_{i=1}^N \bm \omega_i \ind \mathcal N(0,\vV)$, with $\vV= \sum_{t,s=1}^T \bm V_{ts}^0$.\\
\textbf{A\ref{a1}(ii).} By A\ref{a2}(v), $E(\epsilon_{it} c_{it})=0$, and by A\ref{a2}(iii) and (v), $||\epsilon_{it} c_{it}||_{1+\delta/2} \leq  \sup_t \|  c_{it} \|_{2+\delta} \, \sup_t \|\epsilon_{it}\|_{2+\delta}<\infty$. Therefore, by the WLLN, $N^{-1} \sum_{i=1}^N \epsilon_{it} c_{it} \inp 0$. \\
\textbf{A\ref{a1}(iii), (iv), (vi)}. These can be shown by similar arguments to A1(ii).
\end{proof}

\begin{proof}[Proof of Lemma 2] \hfill \\
\textbf{A\ref{a3}(i).} This can be shown by simply applying the WLLN to $\vx_{it} \vx_{is}'$ under A\ref{a1}(i)-(ii) and A\ref{a4}(ii). 
\textbf{A\ref{a3}(ii).} Let $\bm \omega_{ij} = \sum_{t\in I_j^0}(\vx_{it}-\overline{\vx}_{ij})\epsilon_{it}$, independent over $i$ by A\ref{a1}(i). By A\ref{a4}(i), $E(\bm \omega_{ij})=0$. By A\ref{a1}(ii), $||(\vx_{it}-\overline{\vx}_{ij})\epsilon_{it}||_{2+\delta/2}^{2+\delta/2} \leq ||\vx_{it} \epsilon_{it}||_{2+\delta/2}^{2+\delta/2} +  \left(\sum_{t\in I_j^0}||\vx_{it}||_{4+\delta} ||\epsilon_{it}||_{4+\delta}\right)^{2+\delta/2} <\infty$. Therefore, $||\bm \omega_{ij}||_{2+\delta/2} <\infty$. By A\ref{a4}(iii), $E (\bm \omega_{ij} \bm \omega_{ij}') = \sum_{t\in I_j^0} E(\vx_{it} \vx_{it}' \epsilon_{it}^2)- (\Delta T_j^0)^{-1} \sum_{t,s\in I_j^0} E(\vx_{it}\vx_{is}' \epsilon_{it} \epsilon_{is}) = \bm \sum_{t \in I_j^0} \bm N_{jj}^0 - (\Delta T_j^0)^{-1} \bm \sum_{t,s \in I_j^0} \bm N_{js}^0$. Then letting $\bm \omega_i = \vc_{j=1:m^0+1}(\bm \omega_{ij})$, by CLT,  $N^{-1} \sum_{i=1}^N \bm \omega_i \ind \mathcal N(0,  \bm \sum_{t \in I_j^0} \bm N_{jj}^0 - (\Delta T_j^0)^{-1} \bm \sum_{t,s \in I_j^0} \bm N_{js}^0)$. By similar arguments, we can derive the correlation between $\bm \omega_{ij}$ and $\bm \omega_{ik}$, with $k\neq j$, to yield $\vW_1$ as a function of $\bm N_{js}^0$, and obtain the desired result $N^{-1} \sum_{i=1}^N \vc_{j \in \mathcal S} (\bm \omega_j) \ind \mathcal N(0, \vW_1)$.\\
\textbf{A\ref{a3}(iii).} This can be shown by similar arguments to A\ref{a3}(ii).
\end{proof}
\begin{proof}[Proof of Theorem \ref{theo2}]\hfill \\
\textbf{Part (i).  } For notational simplicity, we derive the results for $\mathcal S = \{1,\ldots, m^0+1\}$. The derivation for a smaller set $\mathcal S$ follows by similar arguments. By A\ref{a1}(iv) and A\ref{a3}(i),
\begin{align*}
&N^{-1} \sum_{i=1}^N \sum_{t \in I_j^0} \ (\vx_{it} - \overline \vx_{i,j})(\vx_{it} - \overline \vx_{i,j})' = N^{-1} \sum_{i=1}^N \sum_{t \in I_j^0} \ (\vx_{it} - \overline \vx_{i,j}) \vx_{it}' \\
& = \Delta T_j^0 \vQ_j^0 -  (\Delta T_j^0)^{-1} \sum_{s,t \in I_j^0} N^{-1} \sum_{i=1} \vx_{it} \vx_{is}' +\op(1) = \Delta T_j^0 \vQ_j^0 - (\Delta T_j^0)^{-1} \sum_{j,s \in I_j^0 } \vx_{it} \vx_{is}'+ \op(1) \\
& = \Delta T_j^0\vQ_j^0 - (\Delta T_j^0)^{-1}  \sum_{j,s \in I_j^0} \vOmega_{ts}^0 +\op(1) =\Delta T^0\vQ_j^0 - (\Delta T_j^0)^{-1}  \vQ_{jj}^0 .
\end{align*}
By A\ref{a3}(ii),
$
\vc_{j=1:m^0+1}\left(N^{-1/2} \sum_{i=1}^N \sum_{t \in I_j^0}\ (\vx_{it} - \overline \vx_{i,j})\epsilon_{it} \right) \ind \mathcal N(0, \vW_{1}).
$
Therefore,
\begin{align*}
N^{-1/2} (\hat{\vbeta}_{FE}-\vbeta^0) &= \diag_{j=1,\ldots, m^0+1} \left( N^{-1} \sum_{i=1}^N \sum_{t \in I_j^0} \ (\vx_{it} - \overline \vx_{i,j})(\vx_{it} - \overline \vx_{i,j})' \right) \times \\
& \times \vc_{j=1,\ldots, m^0+1}\left(N^{-1/2} \sum_{i=1}^N \sum_{t \in I_j^0}\ (\vx_{it} - \overline \vx_{i,j})\epsilon_{it} \right)\\
& \ind \diag_{j=1,\ldots, m^0+1} \left[(\Delta T_j^0 \vQ_j^0 - (\Delta T_j^0)^{-1}  (\vQ_{jj}^0) ^{-1} \right] \mathcal N(0, \vW_{1}) = \mathcal N(0, \vOmega_1^{-1} \vW_{1} \vOmega_1^{-1}).
\end{align*}

\textbf{Part (ii). } 
Note that
$
\sqrt{N} (\hat{\vbeta}_{FFE} - \vbeta^0)  =
\left( N^{-1} \sum_{i=1}^N \sum_{t=1}^T \widetilde{\vx}_{it} \widetilde{\vx}_{it}'\right)^{-1}
\left( N^{-1/2} \sum_{i=1}^N \sum_{t=1}^T \widetilde{\vx}_{it} \epsilon_{it}^*\right).
$ 
Recall that $\vw_{ij} = \sum_{t\in I_j^0} \vx_{it}/T$, and let $\vx_{it,j}$ the $(m^0+1)p$ vector with elements $(j-1)p+1, \ldots jp$ equal to $\vx_{it}$, and the rest equal to zero. Also let $ \diag(\bm O_{p\times p},\ldots, \bm O_{p\times p}, \bm A,\bm O_{p\times p}, \ldots \bm O_{p\times p})$ denote the $(m^0+1)p \times (m^0+1)p$ matrix with some diagonal $p\times p$ block equal to $\bm A$, and the rest of the elements including the diagonal $p\times p$ null matrices denoted by $\bm O_{p\times p}$ equal to zero (from the context below, the position of this block is clear). Then for $t \in I_j^0$, 
$
\widetilde \vx_{it}' = [\bm 0_{1\times p}, \ldots,  \bm 0_{1\times p}, \vx_{it,j}', \bm 0_{1\times p} \ldots, \bm 0_{1\times p}'- [\vw_{i1}', \ldots, \vw_{i,m^0+1}']' \equiv \vx_{it,j}' - \vw_i', 
$
with $\bm 0_{1\times p}$ the $1\times p$ null vector, so we have:
\begin{align*}
&N^{-1} \sum_{i=1}^N \sum_{t=1}^T \widetilde{\vx}_{it} \widetilde{\vx}_{it}'  =  \sum_{j=1}^{m^0+1} \sum_{t\in I_j^0} N^{-1} \sum_{i=1}^N (\vx_{it,j} - \vw_i)( \vx_{it,j}-\vw_i)'\\
& = \sum_{j=1}^{m^0+1} \sum_{t\in I_j^0} N^{-1} \sum_{i=1}^N \vx_{it,j}\vx_{it,j}' - \sum_{j=1}^{m^0+1} \sum_{t\in I_j^0} N^{-1} \sum_{i=1}^N \vx_{it,j} \vw_i' \\
&- \sum_{j=1}^{m^0+1} \sum_{t\in I_j^0} N^{-1} \sum_{i=1}^N \vw_i \vx_{it,j}' + \sum_{j=1}^{m^0+1} \sum_{t\in I_j^0} N^{-1} \sum_{i=1}^N \vw_i \vw_i'.
\end{align*}
By A\ref{a1}(iv), $\sum_{t\in I_j^0} N^{-1} \sum_{i=1}^N \vx_{it,j}\vx_{it,j}' \inp \diag(0,\ldots, \Delta T_j^0 \vQ_j^0, \ldots 0)$, so $\sum_{j=1}^{m^0+1} \sum_{t\in I_j^0} N^{-1} \sum_{i=1}^N \vx_{it,j}\vx_{it,j}' \inp \diag_{j=1:m^0+1}(\Delta T_j^0 \vQ_j^0)$, and $ \sum_{j=1}^{m^0+1} \sum_{t\in I_j^0} N^{-1} \sum_{i=1}^N \vx_{it,j} \vw_i' =  - T^{-1}\vc_{j=1:m^0+1}\left(N^{-1} \sum_{i=1}^N \sum_{t\in I_j^0} \vx_{it}\right)$\\$\times \vc_{j=1:m^0+1}\left(N^{-1} \sum_{i=1}^N \sum_{t\in I_j^0} \vx_{it}'\right) = TN^{-1} \sum_{i=1}^N \vw_i \vw_i'$
$\inp - T^{-1}\sum_{j,k=1}^{m^0+1} \vQ_{jk}^0.$ Therefore, \\
$
N^{-1} \sum_{i=1}^N \sum_{t=1}^T \widetilde{\vx}_{it} \widetilde{\vx}_{it}' \inp  \diag_{j=1:m^0+1}(\Delta T_j^0 \vQ_j^0) - (2T^{-1} -T^{-2}) \widetilde \vQ^0,$
the $p(m^0+1) \times p(m^0+1)$ matrix with the $(j,k)$ sub-matrix of size $p\times p$ equal to $\vQ_{jk}^0$. Because $N^{-1} \sum_{i=1}^N \sum_{t=1}^T \widetilde{\vx}_{it} \epsilon_{it}^* \ind \mathcal N(0, \vW_{2})$ by A\ref{a2}(iii), it follows that
$\sqrt{N} (\hat{\vbeta}_{FFE} - \vbeta^0) \ind \mathcal N(0, \vOmega_2^{-1}\vW_{2}\vOmega_2^{-1})$.  

\end{proof}

\begin{proof}[Proof of Theorem \ref{theo3}]\hfill \\
Since $\vOmega_{ts}^0=0$ for $t\neq s$, $\vQ_{jj} = \sum_{t,s \in I_j^0}\vOmega_{ts}^0 = \Delta T_j^0 \vQ_j^0$. Therefore, $\vOmega_{1} = \diag_{j=1:m^0+1} (\Delta T_j^0 \vQ_j^0 - \vQ_j^0)$. Moreover, $\vQ_{jk} = \sum_{t\in I_j^0} \sum_{s \in I_k^0} \vOmega_{ts} = 0$ for $j\neq k$. For $j=k$, $\vQ_{jk} = \sum_{t,s\in I_j^0} \vOmega_{ts} = \Delta T_j^0 Q_j^0$. Therefore $\vOmega_2 = \diag_{j=1:m^0+1} (\Delta T_j^0\vQ_j^0 - (2T^{-1} -T^{-2})\Delta T_j^0 \vQ_j^0) = \diag_{j=1:m^0+1} [\Delta T_j^0 (1-T^{-1})^2\vQ_j^0] = \frac{(T-1)^2}{T^2}\diag_{j=1:m^0+1} (\Delta T_j^0 \vQ_j^0)$. Next,  recall that $$\bm W_{1} = \lim Var \{ \vc_{j=1:m^0+1}[N^{-1/2} \sum_{i=1}^N \sum_{t\in I_j^0} (\vx_{it} - \overline \vx_{ij}) \epsilon_{it}] \}.$$
Because $\vOmega_{ts}^0=0$, and $E(\epsilon_{it} \epsilon_{is}|\vX_i)=0$,
\begin{align*}
& \lim Var [N^{-1/2} \sum_{i=1}^N \sum_{t\in I_j^0} (\vx_{it} - \overline \vx_{ij}) \epsilon_{it}] = E[ N^{-1}  \sum_{i=1}^N \sum_{t,s\in I_j^0} (\vx_{it} - \overline \vx_{ij})(\vx_{it} - \overline \vx_{ij})' \epsilon_{it}^2] \\
& = \sigma^2 \lim [ N^{-1}  \sum_{i=1}^N \sum_{t\in I_j^0} E(\vx_{it} - \overline \vx_{ij})(\vx_{it} - \overline \vx_{ij})' ] = \sigma^2 (\Delta T_j^0 \vQ_j^0 - \vQ_j^0);\\
& \lim E [N^{-1/2} \sum_{i=1}^N \sum_{t\in I_j^0} (\vx_{it} - \overline \vx_{ij}) \epsilon_{it}] [N^{-1/2} \sum_{i=1}^N \sum_{s\in I_k^0} (\vx_{is} - \overline \vx_{ik}) \epsilon_{is}]' =0.
\end{align*} 
Therefore, $\bm W_{1} = \sigma^2 \vOmega_1$, so $\bm V_{FE} = \sigma^2 \vOmega_1^{-1}$. 
Now we calculate $\bm W_{2}$. 
\begin{align*}
\bm W_{2}& = \lim Var [N^{-1/2} \sum_{i=1}^N \sum_{t=1}^T \widetilde \vx_{it}(\epsilon_{it} -\overline \epsilon_i)] = \lim E[N^{-1} \sum_{i=1}^N \sum_{t,s=1}^T \widetilde \vx_{it} \widetilde \vx_{is}' (\epsilon_{it} -\overline \epsilon_i)(\epsilon_{is} -\overline \epsilon_i)] \\
& = \lim N^{-1} \sum_{i=1}^N \sum_{t,s=1}^T E(\widetilde \vx_{it} \widetilde \vx_{is}') E[(\epsilon_{it} -\overline \epsilon_i)(\epsilon_{is} -\overline \epsilon_i)|\vX_i].
\end{align*}
Now, $E[(\epsilon_{it} -\overline \epsilon_i)(\epsilon_{is} -\overline \epsilon_i)|\vX_i] = \sigma^2(1-T^{-1})$ if $t=s$, and $-\sigma^2T^{-1}$ otherwise. Therefore, 
\begin{align*}
\bm W_{2}& =\sigma^2 (1-T^{-1}) \lim N^{-1} \sum_{i=1}^N \sum_{t=1}^T E(\widetilde \vx_{it} \widetilde \vx_{it}') - T^{-1}\sigma^2\lim N^{-1} \sum_{i=1}^N \sum_{t,s=1,s\neq t}^T E(\widetilde \vx_{it} \widetilde \vx_{is}').
\end{align*}
For $t\in I_j^0, s\in I_k^0$ and $t\neq s$, we have:
$
E(\widetilde \vx_{it} \widetilde \vx_{is}')  = E[(\vx_{it,j} - \vw_i)( \vx_{is,k}-\vw_i)'] =  - E(\vx_{it,j} \vw_i') - E[ \vw_i(\vx_{is,k})'] + E(\vw_i \vw_i')$. Therefore,
\begin{align*}
\bm W_{2}& = \sigma^2 \lim N^{-1} \sum_{i=1}^N \left [\sum_{t=1}^T E(\widetilde \vx_{it} \widetilde \vx_{it}')- T^{-1}\sigma^2 \sum_{j,k=1}^{m^0+1} \sum_{t\in I_j^0, s\in I_k^0} E(\vx_{it,j} \vw_i') \right.\\
&\left.- T^{-1}\sigma^2 \sum_{j,k=1}^{m^0+1} \sum_{t\in I_j^0, s\in I_k^0} E(\vw_i \vx_{it,k}') +  T^{-1}\sigma^2 \sum_{t,s=1}^T  E(\vw_i \vw_i') \right]\\
& = \sigma^2 \vOmega_2 - 2 \sigma^2 \sum_{k=1}^{m^0+1} \sum_{s\in I_k^0} E(\vw_i \vw_i') + \sigma^2 T (E\vw_i \vw_i') = \sigma^2 \vOmega_2 - \sigma^2 T E(\vw_i \vw_i') \\
&= \sigma^2 \vOmega_2 -\sigma^2 T^{-1} \diag_{j=1:m^0+1} (\Delta T_j^0 \vQ_j^0) \\
& = \diag_{j=1:m^0+1}\{ \sigma^2 \vQ_j^0 [\Delta T_j^0(1- T^{-1}(2-T^{-1}) -T^{-1}]\} = \frac{T^2-3T+1}{T^2} \diag_{j=1:m^0+1}[ \sigma^2 \vQ_j^0 \Delta T_j^0].
\end{align*}
Therefore, 
$
\bm V_{FFE} = \diag_{j=1:m^0+1}\left[\sigma^2 (\vQ_j^0)^{-1}\frac{(T^2-3T+1) T^2}{(T-1)^4} \frac{1}{\Delta T_j^0}\right] .
$
We now compare it to \\ $\bm V_{FE} = \diag_{j=1:m^0+1}\left[\sigma^2 \frac{1}{\Delta T_j^0-1}(\vQ_j^0)^{-1}\right]$. Since $T\geq 4$, $T^3-5T^2+4T-1 = T(T-1)(T-4) -1 >0$, \begin{align*}
&\frac{(T^2-3T+1) T^2}{(T-1)^4} \frac{1}{\Delta T_j^0} - \frac{1}{\Delta T_j^0-1} = \frac{(\Delta T_j^0 -1)(T^4 - 3T^3 + T^2) - (T-1)^4 \Delta T_j^0}{\Delta T_j^0 (\Delta T_j^0-1)(T-1)^4 } = \frac{\Delta T_j^0 [(T^4 - 3T^3 + T^2)-(T-1)^4] - (T^4 - 3T^3 + T^2)  }{(\Delta T_j^0-1)(T-1)^4} \\
&= \frac{\Delta T_j^0  [T^3-5T^2+4T-1] - (T^4 - 3T^3 + T^2)  }{(\Delta T_j^0-1)(T-1)^4} < \frac{T  [T^3-5T^2+4T-1] - (T^4 - 3T^3 + T^2)  }{(\Delta T_j^0-1)(T-1)^4} =\frac{-T(2T-1)(T-1)}{(\Delta T_j^0-1)(T-1)^4} <0,
\end{align*}
it follows that each of the FFE estimators is strictly more efficient than its FE counterpart, and that $V_{FFE}- V_{FE}$ is negative definite.
\end{proof}

\newpage
\section*{Appendix B}

\begin{table}[h!]                                                                           
        \scriptsize                                                              
        \begin{tabular}{|p{3.5cm}|p{7.5cm}|p{4cm}|}   
                \hline 
                Variable                                                        & Description & Source \\ \hline 
                \vspace{0.2 cm} 1-year house price expectations         & \vspace{0.2 cm}Home owners in the ALP answered the following question: On a scale from 0 percent to 100 percent where 0 means that you think there is a chance and 100 means that you think the event is absolutely sure to happen, what do you think are the chances that by next year at this time your home will be worth more than it is today? \vspace{0.2 cm} & \vspace{0.2 cm}Rand American Life Panel (ALP) \\   
                
                Change in unemployment  & Quartely percentage change in the unemployment rate (on state-level); Changes in unemployment rates are based on data of the past three months. \vspace{0.2 cm} & The state-level unemployment rates are taken from the Bureau of Labor Statistics (http://www.bls.gov). \vspace{0.2 cm} \\    
                
                Change in house prices          & Quartely percentage change in house prices (on state-level); Changes in house prices are based on data of the past two quarters. \vspace{0.2 cm} & The state-level house price index is taken from the Office of Federal Housing Enterprise Oversight (OFHEO). http://www.fhfa.gov \vspace{0.2 cm} \\
                
                Home value                                              & Self-reported house value & ALP \vspace{0.2 cm} \\                    
                
                Income per capita                               & Annual family income divided by the number of household members (see Niu and van Soest 2014 for a detailed description of this variable) \vspace{0.2 cm} & ALP \\               
                
                Household size                                  & Number of household members \vspace{0.2 cm} & ALP \\                       
                
                Married                                                                 & Dummy variable that is equal to one if the home owner is married \vspace{0.2 cm} & ALP \\                          
                
                Health                                                          & Self-reported health status measured on a 1 to 5 scale \vspace{0.2 cm} & ALP  \\                           
                
                Non-economic sentiment  & Survey questions from the ALP on the respondent's life-satisfaction, happiness and wornout are combined into this single measure ranging between 0 and 1 (0 reflects strong dissatisfaction and 1 strong satisfaction). See Niu and van Soest (2014) for a detailed description of this variable. \vspace{0.2 cm} & ALP \\                  
                
                Economic sentiment                      & Survey questions from the ALP on the respondent's satisfaction with their job, total household income, economic and financial situation are combined into this single measure ranging between 0 and 1 (0 reflects strong dissatisfaction and 1 strong satisfaction). See Niu and van Soest (2014) for a detailed description of this variable. \vspace{0.2 cm} & ALP \\                 
                
                Age                                                                     & Age of the respondent in years \vspace{0.2 cm} & ALP \\                                                             
                
                Sand state                                              & Dummy variable that is equal to one if the home owner lives in Arizona, California, Florida or Nevada, the four so-called sand states (these states were most hurt in the real estate collapse) \vspace{0.2 cm} &  \\     
                
                Female                                                          & Dummy variable that is equal to one if the home owner is female \vspace{0.2 cm} & ALP \\                                                        
                
                White                                                           & Dummy variable that is equal to one if the home owner is caucasian \vspace{0.2 cm} & ALP \\                                                            
                
                Schooling                                                       & Dummy variable that is equal to one if the home owner has a Bachelor or higher degree \vspace{0.2 cm} & ALP \\  \hline   
                
        \end{tabular}                        
        \caption{Description of variables used in Section 5.2.}
        \label{tab:datadic}                                                        
\end{table}  

\end{document}